\documentclass[sigplan,screen,dvipsnames]{acmart}
\AtBeginDocument{%
  \providecommand\BibTeX{{%
    \normalfont B\kern-0.5em{\scshape i\kern-0.25em b}\kern-0.8em\TeX}}}

\acmYear{2022}\copyrightyear{2022}
\setcopyright{rightsretained}
\acmConference[PPoPP '22]{27th ACM SIGPLAN Symposium on Principles and Practice of Parallel Programming}{February 12--16, 2022}{Seoul, Republic of Korea}
\acmBooktitle{27th ACM SIGPLAN Symposium on Principles and Practice of Parallel Programming (PPoPP '22), February 12--16, 2022, Seoul, Republic of Korea}
\acmPrice{}
\acmDOI{10.1145/3503221.3508413}
\acmISBN{978-1-4503-9204-4/22/02}

\usepackage[capitalize]{cleveref}

\usepackage{multirow}
\usepackage{subcaption}
\usepackage{colortbl}
\definecolor{ColBlue}{RGB}{227, 252, 252}
\definecolor{ColPink}{RGB}{245, 208, 229}
\definecolor{ColYellow}{RGB}{247, 247, 205}
\newcolumntype{b}{>{\columncolor{ColBlue}}c}
\newcolumntype{f}{>{\columncolor{ColYellow}}c}
\newcolumntype{t}{>{\columncolor{ColPink}}c}

\usepackage{pifont} %
\newcommand{\cmark}{\ding{51}}%
\newcommand{\xmark}{\ding{55}}%

\usepackage{tikz}
\usetikzlibrary{positioning, shapes, fit}

\tikzset{B/.style={draw, inner sep=0pt, circle, font=\footnotesize{}, minimum size=16pt}}
\tikzset{E/.style={draw,->}}

\tikzset{Tag Node/.style={draw, circle, inner sep=1pt}}
\tikzset{Tag Edge/.style={draw, thick}}
\tikzset{Tag Loop/.style={draw, thick, in=65, out=115, looseness=5}}
\newcommand{\TwoTagDepGraph}[3]{
    \begin{tikzpicture}[baseline=(1.base)]
        \node[Tag Node] (1) at (210:0.5) {#1};
        \node[Tag Node] (2) at (330:0.5) {#2};
        #3
    \end{tikzpicture}}

\newcommand{\TopConfigNode}[3]{
    \begin{tabular}{c | c }
        #1 & \{ #2 \} \\
        \hline
        \multicolumn{2}{c}{#3}
    \end{tabular}
}

\newcommand{\ConfigurationNode}[3]{
    node { \TopConfigNode{#1}{#2}{#3} } }

\newcommand{\TopDNode}[4]{
    \TopConfigNode{#1}{#2}{#4}
}

\newcommand{\DNode}[5]{
    node (#5) { \TopDNode{#1}{#2}{#3}{#4} }
}

\usepackage{listings}

\definecolor{PigBlue}{RGB}{30, 0, 200}
\definecolor{PigRed}{RGB}{150, 0, 0}
\definecolor{PigGreen}{RGB}{0, 128, 0}
\lstdefinelanguage{Flumina}
{
    keywords=[1]{
        if, else, for, in, return, output, new, match, with
    },
    keywordstyle=[1]\bfseries,
    keywords=[2]{
        init,
        depends,
        fork, join,
        update, update_i, update_s, update_o,
        out, out_i,
        pred_i, pred_0, pred_j
    },
    keywordstyle=[2]\color{PigRed},
    keywords=[3]{
        Map, Set, List,
        Event, State, Integer, Key, Out, Tag, Payload, Heartbeat,
        State_i, State_j, State_k, State_0, State_1, State_2,
        State1, State2,
        Pred,
        List
    },
    keywordstyle=[3]\color{PigBlue},
    sensitive=true,
    morestring=[b]',
    morecomment=[l][\color{black!70}]{//},
    basicstyle=\scriptsize\ttfamily
}

\lstnewenvironment{FluminaCode}{
  \nopagebreak
  \lstset{language=Flumina}}{}

\newcommand{\fl}[1]{\text{\lstinline[
    language=Flumina,
    columns=fullflexible,
    basicstyle=\footnotesize\ttfamily
]!#1!}}

\newcommand{\heading}[1]{\medskip\noindent\textbf{#1}\enspace}

\newcommand{\tg}[1]{\texttt{#1}} %
\newcommand{\istate}{\texttt{init}}
\newcommand{\fk}{\texttt{f}}
\newcommand{\jn}{\texttt{j}}

\newcommand{\wire}[3]{\ensuremath{\langle #1, #2, #3 \rangle}}
\newcommand{\semantics}[6]{\wire{#1}{#2}{#3} \ensuremath{\underset{#6}{\overset{#4}{\longrightarrow}}} \wire{#1}{#2}{#5}}
\newcommand{\sem}[1]{[\![#1]\!]}

\newcommand{\wfield}[2]{{#1}.{\mathsf{#2}}}
\newcommand{\ttt}[1]{{\small\texttt{#1}}}

\usepackage{tcolorbox}
\newenvironment{takeaway}{
\begin{tcolorbox}[colback=blue!5!white,colframe=blue!5!white,arc=0mm,grow to left by=1.5mm,left=1mm,grow to right by=1.5mm,right=1mm,top=.5mm,bottom=.5mm]
}
{
\end{tcolorbox}
}

\newcommand{\sys}{Flumina}

\begin{document}

\title{Stream Processing with Dependency-Guided Synchronization (Extended Version)}

\author{Konstantinos Kallas}
\authornote{Equal contribution.}
\email{kallas@seas.upenn.edu}
\affiliation{%
  \institution{University of Pennsylvania}
  \city{Philadelphia}
  \state{PA}
  \country{USA}
}

\author{Filip Niksic}
\authornotemark[1]
\authornote{now at Google.}
\email{fniksic@gmail.com}
\affiliation{%
  \institution{University of Pennsylvania}
  \city{Philadelphia}
  \state{PA}
  \country{USA}
}

\author{Caleb Stanford}
\authornotemark[1]
\email{castan@cis.upenn.edu}
\affiliation{
  \institution{University of Pennsylvania}
  \city{Philadelphia}
  \state{PA}
  \country{USA}
}

\author{Rajeev Alur}
\email{alur@cis.upenn.edu}
\affiliation{%
  \institution{University of Pennsylvania}
  \city{Philadelphia}
  \state{PA}
  \country{USA}
}

\begin{abstract}
Real-time data processing applications with low latency requirements have led to the increasing popularity of stream processing systems. While such systems offer convenient APIs that can be used to achieve data parallelism automatically, they offer limited support for computations that require synchronization between parallel nodes. In this paper, we propose \emph{dependency-guided synchronization (DGS)}, an alternative programming model
for stateful streaming computations with complex synchronization requirements. In the proposed model, the input is viewed as partially ordered, and the program consists of a set of parallelization constructs which are applied to decompose the partial order and process events independently. Our programming model maps to an execution model called \emph{synchronization plans} which supports synchronization between parallel nodes. Our evaluation shows that APIs offered by two widely used systems---Flink and Timely Dataflow---cannot suitably expose parallelism in some representative applications. In contrast, DGS enables implementations with scalable performance, the resulting synchronization plans offer throughput improvements when implemented manually in existing systems, and the programming overhead is small compared to writing sequential code.
\end{abstract}

\ccsdesc[500]{Software and its engineering~Parallel programming languages}
\ccsdesc[300]{Software and its engineering~Domain specific languages}
\ccsdesc[300]{Information systems~Stream management}

\keywords{data parallelism, sharding, synchronization, distributed stream processing}

\maketitle

\section{Introduction}
\label{sec:intro}

A wide range of applications in domains such as healthcare,
transportation, and smart homes are increasingly relying on real-time
data analytics with low latency and high throughput requirements.
This has motivated a significant amount of research on \emph{stream
processing}, spanning different layers of abstraction
in the software stack.  At the lowest level, \textbf{\emph{stream
  processing systems}} (e.g.  Flink~\cite{carbone2015flink},
Samza~\cite{Samza2017}, Storm~\cite{Storm}, Spark
Streaming~\cite{DStreams2013}, Trill~\cite{chandramouli2014trill},
Heron~\cite{kulkarni2015twitter-heron},
Beam~\cite{Beam})
handle scheduling, optimizations, and
operational concerns. At the intermediate level, \textbf{\emph{stream processing APIs and programming models}}
(e.g. MapReduce online~\cite{condie2010mapreduce}, SPADE~\cite{gedik2008spade}, SP Calculus~\cite{soule2010universal}, Timely Dataflow~\cite{Timely,murray2013naiad}, StreamIt~\cite{thies2002streamit}, Flink's DataStream API~\cite{Flink}),
usually based on a form of
dataflow~\cite{gilles1974semantics,lee1987synchronous}, abstract the
computation in a way that hides implementation details, while exposing
parallelization information to the underlying system. At the top level, \textbf{\emph{high-level query languages}}
(e.g. Streaming SQL~\cite{jain2008towards,begoli2019one}, SamzaSQL~\cite{pathirage2016samzasql}, Structured Streaming~\cite{armbrust2018structured}, StreamQRE~\cite{mamouras2017streamqre}, CQL~\cite{CQL}, AFAs~\cite{chandramouli2010high})
provide convenient abstractions that are built on top of the streaming APIs.
Of these layers, streaming APIs play a central role in the
successful scaling of applications since their expressiveness restricts
the  available parallelism.  In this paper we focus on
rethinking the dataflow model at this intermediate layer to enable
parallel implementations for a broader range of programs.

The success of stream processing APIs based on the data\-flow model can be attributed to their ability to simplify the task of parallel programming. To accomplish this, most APIs expose a simple but effective model of data-parallelism called \emph{sharding}:
nodes in the dataflow graph are replicated into many parallel instances, each of which
will process a different partition of the input events.
However, while sharding is intuitive for programmers, it also implicitly limits the scope of parallel patterns that
can be expressed. Specifically, it prevents arbitrary
\emph{synchronization across parallel instances}
since it disallows communication between them.
This is limiting in modern applications such as video processing \cite{chienchun2018videoedge} and distributed machine learning \cite{otey2006fast}, since they require both synchronization between nodes and high throughput and could therefore benefit from parallelization.
Further evidence that sharding is limiting in practice
can be found in a collection of feature requests in state-of-the-art stream processing systems~\cite{FLIP8RescalableNonPartitionedStateApacheFlinkApacheSoftwareFoundation-2020-05-27,SEP27SideInputsforLocalStoresApacheSamzaApacheSoftwareFoundation-2020-05-27,KIP114KTablestatestoresandimprovedsemanticsApacheKafkaApacheSoftwareFoundation-2020-05-27}, asking either for state management that goes beyond replication or for some form of communication between shards.
To address these needs, system developers have introduced
extensions to the dataflow model to enable specific use cases such as  \emph{message broadcasting} and \emph{iterative dataflows}.
However, existing solutions do not generalize,
as we demonstrate experimentally in \Cref{ssec:eval-existing-implementations}.
For the remainder of applications, users are left with two unsatisfying solutions: either ignore parallelization potential, implementing their application
with limited parallelism; or circumvent the stream processing APIs using low-level external mechanisms to achieve synchronization between parallel instances.

For example, consider a fraud detection application where the input is a distributed set of streams of bank transaction events. Suppose we want to build an unsupervised online machine learning model over these events which classifies events as fraudulent based on a combination of \emph{local} (stream-specific) and \emph{global} (across-streams) statistical summaries.
The problem with the traditional approach is that when classifying a new event, we need access to both the local and the global summaries; but this cannot be achieved using sharding since by default shards do not have access to a global summary.
One extension to the dataflow model, implemented in some systems~\cite{Flink,Timely} is the \emph{broadcast} pattern, which allows the operator computing the global summary to broadcast to all other nodes.
However, broadcasting is restricted since it does not allow bidirectional communication;
the global summary needs to be both broadcast to all shards, but also updated by all shards.
Cyclic dataflows are another partial solution, but do not always solve the problem, as we show in \Cref{ssec:eval-existing-implementations}.
In practice, applications like this one with complex synchronization requirements opt to  manually
implement the required synchronization using
external mechanisms (e.g. querying a separate a key-value store with strong consistency guarantees). This is error prone and, more importantly, violates the requirements of many streaming APIs that operators need to be effect-free so that the underlying system can provide exactly-once execution guarantees in the presence of faults.

To address the need to combine parallelism with synchronization, we make two contributions. First, we propose \emph{synchronization plans}, a tree-based execution model which is a restricted form of communicating sequential processes~\cite{hoare1978communicating}.
Synchronization plans are hierarchical structures that represent concurrent computation in which parallel nodes are not completely independent, but communicate with their ancestors on special synchronizing events.
While this solves the problem of being able to express synchronizing parallel computations, we still need a streaming API which exposes such parallelism implicitly rather than explicitly.
For this purpose, we propose \emph{dependency-guided synchronization} (DGS), a parallel programming model which can be mapped automatically to synchronization plans.

A DGS program consists of three components. First, the user provides a \emph{sequential implementation} of the computation; this serves to define the semantics of what they want to compute assuming the input is processed as a sequence of events.
Second, the user indicates which input events can be processed in parallel and which require synchronization by providing a \emph{dependence relation} on input events.
This relation induces a partial order on the input stream.
For example, if events can be processed completely in parallel without any synchronization, then all input events can be specified to be independent.
Third, the user provides a mechanism for parallelizing state when the input stream contains independent events: \emph{parallelization primitives} called \emph{fork} and \emph{join}. This model is inspired by classical parallel programming, but has streaming-specific semantics which describes how a partially ordered input stream is decomposed for parallel processing.

Given a DGS program, the main technical challenge is to generate
a synchronization plan, which corresponds to a concrete implementation, that is both \emph{correct} and \emph{efficient}. More precisely, the challenge lies in ensuring that a derived implementation correctly enforces the specified input dependence relation. To achieve correctness, we formalize: (i) a set of conditions that ensure that a program is consistent, and (ii) a notion of $P$-\emph{valid} synchronization plans, i.e., plans that are well-typed with respect to a given program $P$.
To achieve efficiency, we design the framework so that correctness is independent of \emph{which} synchronization plan is chosen---as long as it is $P$-valid.
The idea of this separation is to enable future work on optimized query execution, in which an optimizing component searches for an efficient synchronization plan maximizing a desired cost metric without jeopardizing correctness.
We tie everything together by proving that the end-to-end system is correct,
that is, any concrete implementation that corresponds to an $P$-valid plan is equivalent to a program $P$ that satisfies the consistency conditions.

In order to evaluate DGS, we perform a set of experiments to investigate the data parallelism limitations of Flink~\cite{carbone2015flink}---a representative high-performance stream processing system---and
Timely Dataflow~\cite{murray2013naiad}---a representative system with iterative computation.
We show that these limits can be partly overcome by manually implementing synchronization. However, this comes at a cost: the code has to be specialized
to the number of parallel nodes and similar implementation details, forcing
the user to sacrifice the desirable benefit of \emph{platform independence}.
We then develop \sys{}, an end-to-end prototype that implements DGS in Erlang~\cite{armstrong1993erlang}, and show that it can automatically produce scalable implementations (through generating synchronization plans from the program) independent of parallelism.

This is the extended version of the paper: it includes additional supplementary material as appendices at the end.
In \Cref{appendix:case-studies},
we evaluate \emph{programmability} of DGS via two real-world case studies.
In particular,
we demonstrate that the effort required---as measured by lines of code---to achieve parallelism is minimal compared to the sequential implementation.
In \Cref{appendix:optimizer},
we describe our current optimizer (how we choose a good synchronization
plan), which is based on minimizing communication between nodes.
In \Cref{appendix:correctness},
we provide the parts of the correctness proof which are omitted from the paper.
In \Cref{appendix:flumina-impl},
we provide additional system details in \sys{}:
specifically, we evaluated its latency costs and implemented a basic
state checkpointing mechanism.
Finally, in \Cref{appendix:flumina-code}, \Cref{appendix:timely-code}, and \Cref{appendix:flink-code}, we provide code excerpts of the programs written in \sys{}, Timely, and Flink, respectively, for our evaluation.

In summary, we make the following contributions:
\begin{itemize}
\item
\emph{DGS:} a novel programming model for parallel streaming computations that require synchronization, which allows viewing input streams as partially ordered sets of events. (\Cref{sec:prog-model})
\item
\emph{Synchronization plans:} a tree-based execution model for parallel streaming computations that require synchronization, a framework for generating a synchronization plan given a DGS program, a prototype implementation, and an end-to-end proof of correctness.
(\Cref{sec:dist-impl})
\item
An evaluation that demonstrates: (i) the throughput limits of automatically scaling computations on examples which require synchronization in Flink and Timely; (ii) the throughput and scalability benefits achieved by synchronization plans over such automatically scaling computations; and (iii) the programmability benefits of DGS for synchronization-centered applications
(\Cref{sec:evaluation}).
\end{itemize}
\sys{}, our implementation of DGS, is open-source and available
at \href{https://github.com/angelhof/flumina}{github.com/angelhof/flumina}.

\section{Dependency-Guided Synchronization}
\label{sec:prog-model}

A DGS program consists of three components: a
\emph{sequential implementation}, a \emph{dependence relation} on input
events to enforce synchronization, and \emph{fork} and \emph{join}
parallelization primitives.
In \Cref{ssec:prog-model-correctness} we define program \emph{consistency}, which are requirements on the \emph{fork} and \emph{join} functions to ensure that any parallel implementation generated from the program is equivalent to the sequential one.

\subsection{DGS programs}
\label{ssec:prog-model-walkthrough}

For a simple but illustrative example, suppose that we want to
implement a stream processing application that simulates a map from
keys to counters, in which there are two types of input events: \emph{increment}
events, denoted $\tg{i}(k)$, and \emph{read-reset} events, denoted
$\tg{r}(k)$, where each event has an associated key $k$.  On each
increment event, the counter associated with that key should be
increased by one, and on each read-reset event, the current value of
the counter should be produced as output, and then the counter should be reset to
zero.

\heading{Sequential implementation.}
\label{p:seq-impl}
In our programming model, the user first provides a sequential
implementation of the desired computation.
A pseudocode
version of the sequential implementation for the
example above is shown in \Cref{fig:key-value-store} (left);
Erlang syntax has been edited for readability, and
we use \fl{s[k]} as shorthand for the value associated with the key $k$ in the map
\emph{or} the default value $0$ if it is not present.
It consists of
(i) the state type \fl{State}, i.e. the map from keys to
counters, (ii) the initial value of the state \fl{init},
i.e. an empty map with no keys, and (iii) a function \fl{update},
which contains the logic for processing input events.
Conceptually, the sequential implementation describes how to process
the data assuming it was all combined into a single sequential stream (e.g.,
sorted by system timestamp).  For example, if the input stream
consists of the events $\tg{i}(1), \tg{i}(2), \tg{r}(1), \tg{i}(2),
\tg{r}(1)$, then the output would be $1$ followed by $0$, produced by
the two $\tg{r}(1)$ (read-reset) events.

\begin{figure}[t]
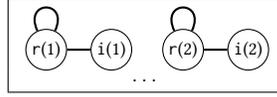

\centering \footnotesize{}
\begin{minipage}{0.38\columnwidth}
\begin{FluminaCode}


// Types
Key = Integer
Event = i(Key) | r(Key)
State = Map(Key, Integer)
Pred = Event -> Bool



// Sequential Code
init: () -> State
init() =
    return emptyMap()
update: (State, Event)
        -> State
update(s, (i(k), ())) =
    s[k] = s[k] + 1;
    return s
update(s, (r(k), ())) =
    output s[k];
    s[k] = 0;
    return s
\end{FluminaCode}
\end{minipage}
\begin{minipage}{0.6\columnwidth}
\begin{FluminaCode}
// Dependence Relation
depends: (Event, Event) -> Bool
depends(r(k1), r(k2)) = k1 == k2
depends(r(k1), i(k2)) = k1 == k2
depends(i(k1), r(k2)) = k1 == k2
depends(i(k1), i(k2)) = false
\end{FluminaCode}

\scriptsize{}
\centering
\begin{tabular}{ | c c | }
    \hline
    \TwoTagDepGraph{\tg{r}(1)}{\tg{i}(1)}{\draw (1) edge[Tag Loop] (1); \draw (1) edge[Tag Edge] (2);}
  & \TwoTagDepGraph{\tg{r}(2)}{\tg{i}(2)}{\draw (1) edge[Tag Loop] (1); \draw (1) edge[Tag Edge] (2);} \\
    \multicolumn{2}{ | c | }{$\cdots$} \\
    \hline
\end{tabular}

\begin{FluminaCode}
// Fork and Join
fork: (State, Pred, Pred)
        -> (State, State)
fork(s, pred1, pred2) =
    // two forked states
    s1 = init(); s2 = init()
    for k in keys(s):
        if pred1(r(k)):
            s1[k] = s[k]
        else:
            // pred2(r(k)) OR
            //   r(k) in neither
            s2[k] = s[k]
    return (s1, s2)
join: (State, State) -> State
join(s1, s2) =
    for k in keys(s2):
        s1[k] = s1[k] + s2[k]
    return s1
\end{FluminaCode}
\end{minipage}

\caption{
  DGS program implementing a map from keys to counters.
  The \fl{depends} relation is visualized as a graph with two keys shown; edges indicate synchronization, while non-edges indicate opportunities for parallelism.
}
\Description{DGS program implementing a map from keys to counters.}
\label{fig:key-value-store}
\end{figure}

\heading{Dependence relation.}
To parallelize a sequential computation, the user needs to provide a
dependence relation which encodes which events are independent, and
thus can be processed in parallel, and which events are dependent, and
therefore require synchronization.  The dependence relation abstractly
captures all the dependency patterns that appear in an application,
inducing a partial order on input events. In this example, there are
two forms of independence we want to expose. To begin with,
\emph{parallelization by key} is possible: the counter map could be
partitioned so that events corresponding to different sets of keys are processed
independently. Moreover, each event is processed atomically in our model,
and therefore \emph{parallelizing increments} on the
counter of the same key is also possible. In particular, different
sets of increments for the same key can be processed independently; we
only need to aggregate the independent counts when a read-reset
operation arrives. On the other hand, read-reset events
are synchronizing for a particular key;
their output is affected by the processing of increments
as well as other read-reset events of that key.

We capture this combination of parallelization and synchronization
requirements by defining the dependence relation
\fl{depends} in \Cref{fig:key-value-store}
(also visualized as a graph)
(see \Cref{ssec:prog-model-formal} for a formal definition).
In the program, the set of events may be \emph{symbolic}
(infinite): here \fl{Event} is parameterized by an integer \fl{Key}.
To allow for this, the dependence relation
is formally a \emph{predicate} on pairs of events,
and is given programmatically as a function from pairs
of \fl{Event} to \fl{Bool}.
For example, \fl{depends(r(k1), r(k2))} (one of four cases)
is given symbolically as equality comparison of keys, \fl{k1 == k2}.
The dependence relation should also be \emph{symmetric},
i.e. \fl{e1} is in \fl{depends(e2)} iff \fl{e2} is in \fl{depends(e1)};
the intuition is that \fl{e1} can be processed in parallel with \fl{e2} iff \fl{e2} can be processed in parallel with \fl{e1}.

\heading{Parallelization primitives: \emph{fork} and \emph{join}.}
\label{p:fork-join}
While the dependence relation indicates the possibility of
parallelization, it does not provide a mechanism for parallelizing
state.  The parallelization is specified using a pair of functions to \fl{fork}
one state into two, and to \fl{join} two states into one.
The
fork function additionally takes as input two predicates of events,
such that the two predicates are \emph{independent} (but not necessarily disjoint):
every event satisfying \fl{pred1} is independent of every event satisfying \fl{pred2}.
The
contract is that after the state is forked into two independent
states, each state will then \emph{only be updated using events satisfying
the given predicate.}
A fork-join pair for our example
is shown in~\Cref{fig:key-value-store}.
The \fl{join} function simply adds up the counts for each key to form the combined
state. The \fl{fork} function has to decide, for each key, which forked state
to partition the count to.
Since read-reset operations \fl{r(k)} are synchronizing,
  i.e., depend on all events of the same key,
  and require knowing the total count,
it partitions by checking
which of the two forked states is responsible for processing read-reset
operations, if any.

The programming model \emph{exposes} parallelism,
but the implementation (\Cref{sec:dist-impl}) determines when to call forks and
joins.
To do this, the implementation instantiates a synchronization plan:
  a tree structure where each node is a stateful worker with a predicate indicating the set of events that it is responsible for.
Nodes that do not have an ancestor-descendant relationship process independent but not necessarily disjoint sets of events.
When a node with children needs to process an event,
  it first uses \fl{join} to merge the states of its children,
  and then it \fl{fork}s back its children states using the combined predicates of its descendants, \fl{pred1} for the left subtree, and \fl{pred2} for the right subtree.
The implementation can therefore instantiate synchronization plans with different shapes and predicates to enable different kinds of parallelism.
For example, to indicate
\emph{parallelization by key},
the left child with \fl{pred1} might contain all events of key $1$
and the right child with \fl{pred2} might contain all events of key $2$.
On the other hand, to indicate \emph{parallelization on increments},
\fl{pred1} and \fl{pred2} might both contain \fl{i(3)}, and in this case neither would contain \fl{r(3)} (to satisfy the independence requirement).
The latter example also emphasizes that \fl{pred1} and \fl{pred2} need not be disjoint, nor need they collectively cover all events.
For the events not covered, in this case \fl{r(3)}, a join would need to be called before an \fl{r(3)} event can be processed.
Parallelization can also be done repeatedly;
the fork function can be called again on a forked state to fork it into two sub-states, and each time the predicates \fl{pred1} and \fl{pred2} will be even further restricted.

\subsection{Formal definition}
\label{ssec:prog-model-formal}

A DGS program can be more general than we have discussed so far,
because we allow for multiple \emph{state types}, instead of just
one. The initial state must be of a certain type, but forks and joins
can convert from one state type to another: for example, forking a
pair into its two components.
Additionally,
each state type can come with a \emph{predicate} which restricts
the allowed events processed by a state of that type.
The complete programming model is summarized in the following
definition.

\begin{definition}[DGS program]
\label{def:prog-model}
Let \fl{Pred(T)} be a given type of \emph{predicates}
on a type \fl{T},
where predicates can be evaluated as functions \fl{T -> Bool}.
A program consists of the following components:
\begin{enumerate}
\item A type of input events \fl{Event}.
\item The dependence relation \fl{depends: Pred(Event, Event)},
which is symmetric: \fl{depends(e1, e2)}
iff \fl{depends(e2, e1)}.
\item A type for output events \fl{Out}.
\item Finitely many state types \fl{State_0}, \fl{State_1}, etc.
\item For each state type \fl{State_i},
a predicate which specifies which input values this type of state can process,
denoted \fl{pred_i: Pred(Event)}. We require \fl{pred_0} $=$ \fl{true}.
\item A sequential implementation, consisting of a single initial state \fl{init: State_0} and for each state type \fl{State_i},
a function \fl{update_i:} \fl{(State_i, Event) -> State_i}.
The update also produces zero or more outputs, given by a function
\fl{out_i: (State_i, Event) -> List(Out)}.
\item A set of parallelization primitives, where each is either a \emph{fork} or a \emph{join}. A fork has type
\[
\fl{(State_i, Pred(Event), Pred(Event)) -> (State_j, State_k)},
\]
and a join has type
\fl{(State_j, State_k) -> State_i}, for some $i, j,$ and $k$.
\end{enumerate}
\end{definition}

\begin{figure}
\centering \scriptsize{}
\scalebox{0.8}{
\begin{tikzpicture}[node distance=0.9cm and 0.9cm, on grid]
\node[B] (0) {$\istate$};
\node[B,right = of 0] (1) {${\tg{r}(1)}$};
\node[B,right = of 1] (2) {$\fk$};
\node[B,below right = of 2] (2b) {${\tg{i}(1)}$};
\node[B,right = of 2b] (3) {$\fk$};
\node[B,above = of 3] (2a) {${\tg{i}(1)}$};
\node[B,below right = of 3] (3b) {${\tg{i}(1)}$};
\node[B,above right = of 3b] (4) {$\jn$};
\node[B,above right = of 4] (5) {$\jn$};
\node[B,right = of 5] (6) {${\tg{r}(1)}$};
\coordinate[right=0.5cm of 6] (7);
\node[B,above = of 0] (t0) {$\istate$};
\node[B,right = of t0] (t1) {${\tg{r}(1)}$};
\node[B,above right = of 2] (t2) {${\tg{i}(1)}$};
\node[B,right = of t2] (t3) {${\tg{i}(1)}$};
\node[B,right = of t3] (t4) {${\tg{i}(1)}$};
\node[B,above = of 6] (t5) {${\tg{r}(1)}$};
\coordinate[right=0.5cm of t5] (t6);
\draw[E] (0) -- (1);
\draw[E] (1) -- (2);
\draw[E] (2) -- (2a);
\draw[E] (2) |- (2b);
\draw[E] (2b) -- (3);
\draw[E] (3) -- (4);
\draw[E] (3) |- (3b);
\draw[E] (3b) -| (4);
\draw[E] (4) -| (5);
\draw[E] (2a) -- (5);
\draw[E] (5) -- (6);
\draw[E] (6) -- (7);
\draw[E] (t0) -- (t1);
\draw[E] (t1) -- (t2);
\draw[E] (t2) -- (t3);
\draw[E] (t3) -- (t4);
\draw[E] (t4) -- (t5);
\draw[E] (t5) -- (t6);
\end{tikzpicture}
}
\caption{Example of a sequential (top) and parallel (bottom) execution of the program in \Cref{fig:key-value-store} on the input stream $\tg{r}(1), \tg{i}(1), \tg{i}(1), \tg{i}(1), \tg{r}(1)$
($\fk$ and $\jn$ denote forks and joins).
}
\Description{Example of a sequential (top) and parallel (bottom) execution of the DGS program implementing a map from keys to counters.}
\label{fig:example-wire-diagram}
\end{figure}
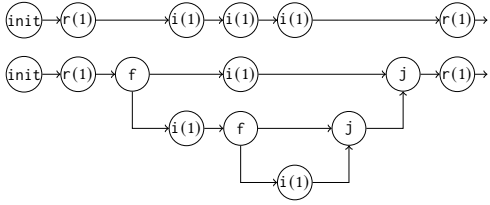

\heading{Semantics.}
The semantics of a program
can be visualized using wire
diagrams, as in \Cref{fig:example-wire-diagram}. Computation proceeds from left to right. Each wire is associated with
(i) a state (of type \fl{State_i} for some $i$) and
(ii) a predicate (of type \fl{Pred(Event)})
which restricts the input events that this wire can process.
Input events are processed as
updates to the state, which means they take one input wire and produce
one output wire, while forks take one input wire and produce two, and
joins take two input wires and produce one. Notice that the same updates are present in both sequential and parallel executions. It is guaranteed in the
parallel execution that \emph{fork and join come in pairs}, like
matched parentheses. Each predicate that is given as input to the \fl{fork} function indicates the set of input events that can be processed along one of the outgoing wires.
Additionally, we require that updates on parallel wires must be
on independent events.
In the example,
the wire is forked into two parts and then forked again, and all three
resulting wires process $\tg{i}(1)$ events. Note that $\tg{r}(1)$
events cannot be processed at that time because they are dependent on
$\tg{i}(1)$ events.
More specifically, we require that
the predicate at each wire of type \fl{State_i}
implies \fl{pred_i}, and that
after each \fl{fork} call,
the predicates at each resulting wire denote independent sets of events.
This semantics is formalized in the following definition.

\begin{definition}[DGS Semantics]
\label{def:prog-model-semantics}
A \emph{wire} is a triple written
using the notation
\wire{\fl{State_i}}{\fl{pred}}{\fl{s}},
where
\fl{State_i} is a state type,
\fl{s: State_i},
and \fl{pred: Pred(Event)}
is a predicate
such that \fl{pred} implies \fl{pred_i}.
We give the semantics of a program through an
inductively defined relation, which we denote
\semantics{\fl{State}}{\fl{pred}}{\fl{s}}{\fl{u}}{\fl{s'}}{\fl{v}},
where
\wire{\fl{State}}{\fl{pred}}{\fl{s}} and \wire{\fl{State}}{\fl{pred}}{\fl{s'}}
are the starting and ending wires
(with the same state type and predicate),
\fl{u: List(Event)} is an input stream,
and \fl{v: List(Out)} is an output stream.
Let \fl{l1 + l2} be list concatenation, and
define \fl{inter(l, l1, l2)} if \fl{l} is some interleaving of \fl{l1} and \fl{l2}.
For \fl{e1}, \fl{e2: Event},
let \fl{indep(e1, e2)} denote that
\fl{e1} and \fl{e2} are not dependent, i.e.
\fl{not(depends(}{\fl{e1}}\fl{,}{\fl{e2}}\fl{))}.
There are two base cases and two inductive cases.
(1)
For any \fl{State}, \fl{pred}, \fl{s},
\semantics{\fl{State}}{\fl{pred}}{\fl{s}}{\fl{[]}}{\fl{s}}{\fl{[]}}.
(2)
For any \fl{State}, \fl{pred}, \fl{s}, and any \fl{e: Event},
if \fl{e} satisfies \fl{pred} then
\semantics{\fl{State}}{\fl{pred}}{\fl{s}}{\fl{[e]}}{\fl{update(s, e)}}{\fl{out(s, e)}}.
(3)
For any \fl{State}, \fl{pred}, \fl{s}, \fl{s'}, \fl{s''},
\fl{u}, \fl{v}, \fl{u'}, and \fl{v'}, if
\semantics{\fl{State}}{\fl{pred}}{\fl{s}}{\fl{u}}{\fl{s'}}{\fl{v}}
and
\semantics{\fl{State}}{\fl{pred}}{\fl{s'}}{\fl{u'}}{\fl{s''}}{\fl{v'}},
then
\semantics{\fl{State}}{\fl{pred}}{\fl{s}}{\fl{u + u'}}{\fl{s''}}{\fl{v + v'}}.
(4)
Lastly, for any instances of
\fl{State}, \fl{State1}, \fl{State2},
\fl{pred}, \fl{pred1}, \fl{pred2},
\fl{s},
\fl{s1'}, \fl{s2'},
\fl{u}, \fl{u1}, \fl{u2}, \fl{v}, \fl{v1}, \fl{v2},
\fl{fork}, and \fl{join},
suppose that
(the conjunction) \fl{pred1(e1)} and \fl{pred2(e2)} implies
\fl{indep(e1, e2)},
\fl{pred1} implies \fl{pred}, and
\fl{pred2} implies \fl{pred}.
Let
\fl{fork(s, pred1, pred2) =} \fl{(s1, s2)}
and
\fl{join(s1', s2') = s'}.
If we have
\fl{inter(u, u1, u2)},
\fl{inter(v, v1, v2)},
\semantics{\fl{State1}}{\fl{pred1}}{\fl{s1}}{\fl{u1}}{\fl{s1'}}{\fl{v1}},
and
\semantics{\fl{State2}}{\fl{pred2}}{\fl{s2}}{\fl{u2}}{\fl{s2'}}{\fl{v2}},
then
\semantics{\fl{State}}{\fl{pred}}{\fl{s}}{\fl{u}}{\fl{s'}}{\fl{v}}.

Finally, the \emph{semantics} $\sem{P}$ of the program $P$
is the set of pairs $(\fl{u}, \fl{v})$
of an input stream \fl{u} and an output stream \fl{v} such that
\semantics{\fl{State_0}}{\fl{true}}{\fl{init}}{\fl{u}}{\fl{s'}}{\fl{v}}
for some \fl{s'}.
\end{definition}

\heading{Representing predicates.}
In the running example, a predicate on a type \fl{T} was represented
as a function \fl{T -> Bool}, but note that the programming model above allows other
representation of predicates, for example using logical formulas.
The tradeoff here is that a more general representation allows
more dependence relations to be expressible, but also complicates the
implementation of an appropriate \fl{fork} function as it must accept
as input more general input predicates.
In our implementation (see \Cref{sec:dist-impl}), we assume
that an event consists of a pair of a \fl{Tag} (relevant for parallelization) and a \fl{Payload} (used only for processing),
where predicates are given as sets of tags (or pairs of tags, for \fl{depends}).
This allows simpler logic in the fork function whose input predicates
are then \fl{Tag -> Bool} and don't depend on the irrelevant payload.
In our example, $\tg{i}(k)$ or $\tg{r}(k)$ would be tags (not payload)
as they are relevant for parallelization.

\subsection{Consistency conditions}
\label{ssec:prog-model-correctness}

Any parallel execution
is guaranteed to preserve the sequential semantics, i.e. processing
all input events \emph{in order} using the \fl{update} function,
as long as the following \emph{consistency conditions} are
satisfied.
The sufficiency of these conditions is shown in
\Cref{thm:consistency-implies-determinism}, which states
that \emph{consistency implies determinism up to output reordering}.
This is a key step in the end-to-end proof of correctness in \Cref{ssec:proof-of-correctness}.
Consistency can be thought of as analogous to the commutativity and
associativity requirements for a MapReduce program to have
deterministic output~\cite{dean2008map-reduce}: just as with MapReduce
programs, the implementation does \emph{not} assume the conditions are
satisfied, but if not the semantics will be dependent on how the
computation is parallelized.

\begin{definition}[Consistency]
\label{def:prog-model-consistency}
A program is \emph{consistent} if the following equations always hold:
\begin{align*}
\fl{join(update(s1,e),s2)}
    &= \fl{update(join(s1,s2),e)} \tag{C1} \\
\fl{join(fork(s,pred1,pred2))}
    &= \fl{s} \tag{C2} \\
\fl{update(update(s,e1),e2))}
    &= \fl{update(update(s,e2),e1))} \tag{C3}
\end{align*}
subject to the following additional qualifications.
First, equation (C1)
is over all
joins \fl{join: (State_j, State_k) ->} \fl{State_i},
events \fl{e: Event} such that \fl{pred_i(e)} and \fl{pred_j(e)},
and states \fl{s1: State_j}, \fl{s2: State_k},
where \fl{update} denotes the update function on the appropriate type.
Additionally the corresponding output on both sides must be the same:
$\fl{out(s1, e)} = \fl{out(join(s1, s2))}$.
Equation (C2)
is over all fork functions
\fl{fork:} \fl{(State_i,Pred(Event),}\fl{Pred(Event)) -> (State_j, State_k)},
all joins \fl{join: (State_j, State_k) -> State_i},
states \fl{s: State_i}, and predicates \fl{pred1} and \fl{pred2}.
Equation (C3)
is over all state types \fl{State_i},
states \fl{s: State_i},
and pairs of \emph{independent} events
\fl{indep(e1, e2)}
such that \fl{pred_i(e1)} and \fl{pred_i(e2)}.
As with (C1), we also require that the outputs
on both sides agree:
\begin{align*}
&\fl{out(s, e1) + out(update(s, e1), e2)} \\
&\; = \; \fl{out(update(s, e2), e1) + out(s, e2)}.
\end{align*}
\end{definition}

Let us illustrate the consistency conditions for our running example
(\Cref{fig:key-value-store}).
If \fl{e} is an increment event,
then condition (C1) captures the fact that counting can be done in parallel:
it reduces to $\fl{(s1[k] + s2[k]) + 1} = \fl{(s1[k] + 1) + s2[k]}$.
Condition (C2) captures the fact that we preserve total count
across states when forking: it reduces to $\fl{s[k] + 0} = \fl{s[k]}$.
Condition (C3) would not be valid for \emph{general} events
\fl{e1}, \fl{e2}, because a read-reset event does not commute with an increment
of the same key ($\fl{s[k] + 1} \ne \fl{s[k]}$),
hence the restriction that \fl{indep(e1, e2)}.
Finally, one might think that a variant of (C1) should hold for
\fl{fork} in addition to \fl{join},
but this turns out not to be the case:
for example, starting from $\fl{s[k] = 100}$,
an increment followed by a fork might yield the pair of counts
$(101, 0)$, while a fork followed by an increment might yield
$(100, 1)$.
It turns out however that commutativity only with joins, and not with
forks, is enough to imply \Cref{thm:consistency-implies-determinism}.

\begin{theorem}
\label{thm:consistency-implies-determinism}
If $P$ is consistent, then $P$ is deterministic up to output reordering.
That is,
for all $(\fl{u}, \fl{v}) \in \sem{P}$,
the multiset of events in stream $\fl{v}$
is equal to the multiset of events in
$\fl{spec(u)}$
where $\fl{spec}$ is the semantics of the sequential implementation.
\end{theorem}
\begin{proof}
We show by induction on the semantics in \Cref{def:prog-model-semantics}
that every wire diagram is equivalent (up to output reordering) to the sequential sequence
of updates.
The sequential inductive step (3) is direct by associativity of function composition
on the left and right sequence of updates (no commutativity of updates is required).
For the parallel inductive step (4), we replace the two parallel wires with sequential wires,
then apply (C1) repeatedly on the last output to move it outside of the parallel wires,
then finally apply (C2) to reduce the now trivial parallel wires to a single wire.
\end{proof}

\section{Synchronization Plans}
\label{sec:dist-impl}

In this section we describe \emph{synchronization plans},
which represent streaming program implementations, and our
framework for generating them from the given DGS program
in \Cref{sec:prog-model}.
Generation of an implementation can be
conceptually split in two parts, the first ensuring correctness and
the second affecting performance. First a program $P$ induces a set of
$P$-\emph{valid}, i.e. correct with respect to it, synchronization
plans. Choosing one of those plans is then an independent optimization
problem that does not affect correctness and can be delegated to a
separate optimization component (\Cref{ssec:optimization-problem}).
Finally, the workers in synchronization plans
need to process some incoming events in order while some can be
processed out of order (depending on the dependence relation).  We
propose a selective reordering technique
(\Cref{ssec:runtime}) that can be used in tandem with heartbeats to
address this ordering issue.
We tie everything together by providing an end-to-end proof that the
implementation is correct with respect to a consistent program
$P$ (and importantly, independent of the synchronization plan chosen
as long as it is $P$-valid) in \Cref{ssec:proof-of-correctness}.
Before describing the separate framework components, we first
articulate the necessary underlying assumptions about input streams in
\Cref{ssec:distributed-assumptions}.

\subsection{Preliminaries}
\label{ssec:distributed-assumptions}

In our model the input is partitioned in some number of input streams that could be distributed, i.e. produced at different locations.
We assume that the implementation has access to \emph{some} ordering
relation $\mathcal{O}$ on pairs of input events (also denoted
$<_{\mathcal{O}}$), and the order of events is increasing along each
input stream. This is necessary for cases where the
\emph{user-written} program requires that events arriving in
different streams are dependent, since it allows the implementation to
have progress and process these dependent events in order.
Concretely, in our setting $\mathcal{O}$ is implemented using \emph{event timestamps}.
Note that these timestamps do not need to correspond to real time, if this is not required by the application.
In cases where real-time timestamps are required, this can be achieved with well-synchronized clocks, as has been done in other systems, e.g. Google Spanner~\cite{corbett2013spanner}.

Each event in each input stream is given by a quadruple $\langle tg, id, ts, v \rangle$,
where $tg$ is a \emph{tag} used for parallelization, $id$ is a unique identifier of the input stream, $ts$ is a \emph{timestamp}, and $v$ is a \emph{payload}.
Of these, only the tag and payload are visible to the programming model
in \Cref{sec:prog-model}, and only the tag is used in predicates and
in the dependence relation.
Our implementation currently requires that the number of possible tags $tg$ is finite (e.g. up to some maximum key) as well as the number of identifiers $id$.

For the rest of this section,
we write events as $\langle \sigma, ts, v \rangle$
where the pair $\sigma = \langle tg, id \rangle$ is called the \emph{implementation tag}.
This is a useful distinction because at the implementation level,
these are the two components that are used for parallelization.
The
relation \fl{depends: (Tag, Tag) -> Bool} in the program straightforwardly lifts to predicates over tags and to implementation tags.

\subsection{Synchronization plans}
\label{ssec:distributed-configurations}

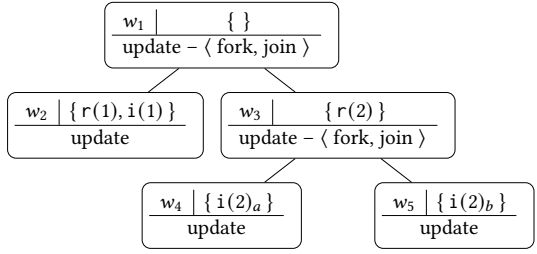
\begin{figure}
\scalebox{0.8}{
\begin{tikzpicture}[sibling distance=11em,
  every node/.style = {shape=rectangle,
    rounded corners,
    draw, align=center}]]
  \node { \TopConfigNode{$w_1$}{}{update -- $\langle$ fork, join $\rangle$} }
    child { \ConfigurationNode{$w_2$}{$\tg{r}(1), \tg{i}(1)$}{update} }
    child { \ConfigurationNode{$w_3$}{$\tg{r}(2)$}{update -- $\langle$ fork, join $\rangle$}
        child { \ConfigurationNode{$w_4$}{$\tg{i}(2)_a$}{update} }
        child { \ConfigurationNode{$w_5$}{$\tg{i}(2)_b$}{update} } };
\end{tikzpicture}
}
\caption{Example synchronization plan derived from the
  program in \Cref{fig:key-value-store} for two keys $k=2$ and
  five input streams $\tg{r}(1), \tg{i}(1), \tg{r}(2), \tg{i}(2)_a, \tg{i}(2)_b$.
  Implementation tags $\tg{i}(2)_a, \tg{i}(2)_b$ both correspond to $\tg{i}(2)$ events but are separate because they arrive in different input streams.
  }
\Description{Example synchronization plan.}
\label{fig:example-configuration}
\end{figure}

Synchronization plans are binary tree structures that encode (i)
parallelism: each node of the tree represents a sequential thread of
computation that processes input events; and (ii) synchronization:
parents have to synchronize with their children to process an event.
Synchronization plans are inspired by prior work on concurrency
models including fork-join concurrency~\cite{frigo1998implementation,lea2000java} and
CSP~\cite{hoare1978communicating}. An example synchronization plan for
the program in \Cref{fig:key-value-store} is shown in
\Cref{fig:example-configuration}.  Each node has an id $w_i$, contains
the set of implementation tags that it is responsible for, a state
type (which is omitted here since there is only one state type
\texttt{State}), and a triple of update, fork, join functions.
Note that a node is responsible to process events from its set of
implementation tags, but can potentially handle all the implementation
tags of its children. The leaves of the tree can process events
independently without blocking, while parent nodes can only process an
input event if their children states are joined. Nodes without a
ancestor-descendant relationship do not directly communicate, but
instead learn about each other when their common ancestor
joins and forks back the state.

\begin{definition}[Synchronization Plans]
  Given a program $P$,
  a synchronization plan is a pair $(\overline{w}, \mathrm{par})$,
  which includes a set of workers $\overline{w} = \{ w_1, \ldots, w_N \}$,
  together with a parent relation $\mathrm{par} \subseteq \overline{w} \times \overline{w}$,
    the transitive closure of which is an ancestor relation denoted as
    $\mathrm{anc} \subseteq \overline{w} \times \overline{w}$.
  Workers have three components:
    (i) a state type $\wfield{w}{state}$ which references one of the state types of $P$,
    (ii) a set of implementation tags $\wfield{w}{itags}$ that the worker is responsible for,
    and (iii) an update $\wfield{w}{update}$ and possibly a fork-join pair $\wfield{w}{fork}$ and $\wfield{w}{join}$ if it has children.
\end{definition}

We now define what it means for a synchronization plan to be $P$-\emph{valid}.
Intuitively, an $P$-valid plan is well-typed with respect to program $P$,
  and the workers that do not have an ancestor-descendant relationship should handle independent and disjoint implementation tags.
$P$-validity is checked syntactically by our framework
  and is a necessary requirement to ensure that the generated implementation is correct (see \Cref{ssec:proof-of-correctness}).

\begin{definition}[$P$-valid]
Formally, a $P$-valid plan has to satisfy the following syntactic properties:
(V1) The state $\fl{State_i} = \wfield{w}{state}$ of each worker $w$ should be
consistent with its update-fork-join triple and its implementation
tags.
The update must be defined on the node state type,
  i.e., $\wfield{w}{update} : \fl{(State_i, Event) -> State_i}$,
    $\fl{State_i}$ should be able to handle the tags corresponding to $\wfield{w}{itags}$,
    and the fork-join pair should be defined for the state types of the node and its children.
(V2) Each pair of nodes that do not have an ancestor-descendant relation, should handle pairwise independent and
disjoint implementation tag sets,
  i.e., $\forall w, w' \not \in \mathrm{anc}(w, w'),
  \wfield{w}{itags} \cap \wfield{w'}{itags} = \emptyset \wedge
  \mathrm{indep}(\wfield{w}{itags}, \wfield{w'}{itags})$.
\end{definition}

\noindent
As an example, the synchronization plan shown in \Cref{fig:example-configuration} satisfies both properties;
    there is only one state type that handles all tags and implementation tag sets are disjoint for ancestor-descendants.
The second property (V2) represents the main idea behind our execution model;
    independent events can be processed by different workers without communication.
Intuitively, in the example in \Cref{fig:example-configuration}, by assigning the
responsibility for handling tag $\tg{r}(2)$ to node $w_3$, its children can
independently process tags $\tg{i}(2)_a, \tg{i}(2)_b$ that are dependent on
$\tg{r}(2)$.

\subsection{Optimization problem}
\label{ssec:optimization-problem}

As described in the previous section, a set of $P$-valid synchronization
plans can be derived from a DGS program $P$. This decouples
the optimization problem of finding a well-performing implementation,
allowing it to be addressed by an independent optimizer,
which takes as input a
description of the available computer nodes and the input streams. This
design means that different optimizers could be implemented for
different performance metrics (e.g. throughput, latency, network load,
energy consumption).
The design space for optimizers is vast and thoroughly exploring it is outside of the scope of this work. For evaluation purposes, we have implemented a few simple optimizers, one of which tries to minimize communication between workers by placing them close to their inputs. Intuitively, it searches for divisions of the event tags into two sets such that those sets are ``maximally'' independent, using those sets of tags for the initial fork, and then recursing on each independent subset.
Its design is described in more detail in \Cref{appendix:optimizer}.

\subsection{Implementation}
\label{ssec:runtime}

Each node of the synchronization plan can be separated into two
components: an event-processing component
responsible for
executing \fl{update},
\fl{fork}, and \fl{join} calls; and a mailbox
component responsible for enforcing ordering requirements and
synchronization.

\heading{Event processing.}
The worker processes execute the functions (\fl{update},
\fl{fork}, and \fl{join}) associated
with the tree node.  Whenever a worker is handed a message
by its mailbox, it first checks if it has any active children, and if
so, it sends them a join request and waits until it receives their
responses. After receiving these responses, it executes the join
function to combine their states, executes the update function on the
received event, and then executes the fork function on the new state,
and sends the resulting states to its children. In contrast, a leaf
worker just executes the update function on the received
event.

\heading{Event reordering.}
The mailbox of each worker ensures that it processes incoming
dependent events in the correct order by implementing the following
selective reordering procedure. Each mailbox contains an event buffer
and a timer for each implementation tag. The buffer holds the events
of a specific tag in increasing order of timestamps and the timer
indicates the latest timestamp that has been encountered for each tag.
When a mailbox receives an event $\langle \sigma, ts, v \rangle$ (or a
join request), it follows the procedure described below. It first
inserts it in the corresponding buffer and updates the timer for
$\sigma$ to the new timestamp $ts$. It then initiates a
cascading process of releasing events with tags $\sigma'$ that depend
on $\sigma$. During that process all dependent tags $\sigma'$ are
added to a dependent tag workset, and the buffer of each tag in the
workset is checked for potential events to release. An event $e =
\langle \sigma, ts, v \rangle$ can be released to the worker process
if two conditions hold. The timers of its dependent tags are higher
than the timestamp $ts$ of the event (which means that the mailbox has
already seen all dependent events up to $ts$, making it safe to
release $e$), and the earliest event in each buffer that $\sigma$
depends on should have a timestamp $ts' > ts$ (so that events are
processed in order). Whenever an event with tag $\sigma$ is released,
all its dependent tags are added to the workset and this process
recurses until the tag workset is empty.

\heading{Heartbeats.}
\label{ssec:heartbeats}
As discussed in \Cref{ssec:distributed-assumptions}, a dependence
between two implementation tags $\sigma_1$ and $\sigma_2$ requires the
implementation to process any event $\langle \sigma_1, t_i, v_i
\rangle$ after processing all events $\langle \sigma_2, t_j, v_j
\rangle$ with $t_j \leq t_i$. However, with the current assumptions on
the input streams, a mailbox has to wait until it receives the
earliest event $\langle \sigma_2, t_j, v_j \rangle$ with $t_j > t_i$,
which could arbitrarily delay event processing. We address this issue
by periodically generating
\emph{heartbeat} events
at each producer,
which are system events that represent
the absence of events on a stream.
Heartbeats are interleaved together with standard events of input streams.
When a heartbeat event $\langle \sigma, t \rangle$ first enters the system,
  it is broadcast to all the worker processes that are descendants of the worker that is responsible for tag $\sigma$.
Each mailbox that receives the heartbeat updates its timers and clears its buffers
  as if it has received an event of $\langle \sigma, t, v \rangle$ without adding the heartbeat to the buffer to be released to the worker process.
Similar mechanisms are often used in other stream processing systems under various names,
  e.g. heartbeats~\cite{heartbeats2005},
  punctuation~\cite{punctuation2003},
  watermarks~\cite{carbone2015flink}, or
  pulses~\cite{schneider2015safeparallelism}.

In our experience, heartbeat rates are successful in improving the latency of the system unless they are configured to be very large or very low values. For a wide range of heartbeat values($\sim$10-1000 per synchronization event), the performance of the system is stable and exhibits minor variance---see more details in \Cref{appendix:flumina-impl}.

\subsection{Proof of correctness}
\label{ssec:proof-of-correctness}

We show that \emph{any} implementation produced by the end-to-end
framework is correct according to the semantics of the programming model (\Cref{theorem:correctness}).
First, \Cref{def:valid-input-instance} formalizes the
assumptions about the input streams outlined in
\Cref{ssec:distributed-assumptions}, and \Cref{def:distr-correctness}
defines what it means for an implementation to be correct with respect
to a sequential specification.
Our definition is inspired by the
classical definitions of distributed correctness based on
observational trace semantics (e.g.,~\cite{lynch1996distributed}).
However, we focus on how to interpret the independent input streams as
a sequential input, in order to model possibly synchronizing and
order-dependent stream processing computations.
The proof of \Cref{theorem:correctness} can be found in \Cref{appendix:correctness}.

\begin{definition}
\label{def:valid-input-instance}
A \emph{valid input instance} consists of
$k$ input streams (finite sequences) \fl{u_1}, \fl{u_2}, \ldots{}, \fl{u_k} of type \fl{List(Event |} \fl{Heartbeat)},
and an order relation $\mathcal{O}$ on input events and heartbeats, with the following properties.
(1) \emph{Monotonicity:} for all $i$, \fl{u_i} is in strictly increasing order according to $\mathrel{<_{\mathcal{O}}}$.
(2) \emph{Progress:} for all $i$, for each input event (non-heartbeat) \fl{x} in \fl{u_i},
for every other stream $j$ there exists an event or heartbeat \fl{y} in \fl{u_j} such that ${\fl{x}} \mathrel{<_{\mathcal{O}}} {\fl{y}}$.
\end{definition}

\noindent
Given a DGS program $P$,
the \emph{sequential specification} is a function \fl{spec:} \fl{List(Event) -> List(Out)}, derived from the sequential implementation
by applying only \fl{update} and no \fl{fork} and \fl{join} calls.
The output specified by \fl{spec} is produced incrementally (or \emph{monotonically}): if \fl{u} is a prefix of \fl{u'}, then \fl{spec(u)} is a subset of \fl{spec(u')}.
Define the \emph{sort} function
${\fl{sort}}_{\mathcal{O}}:$ \fl{List(List(Event |} \fl{Heartbeat)) ->} \fl{List(Event)}
which takes $k$ sorted input event streams and sorts them
into one sequential stream, according to the total order relation $\mathcal{O}$, and drops heartbeat events.

\begin{definition}
\label{def:distr-correctness}
A distributed implementation is \emph{correct} with respect to a given sequential specification \fl{spec:} \fl{List(Event) ->} \fl{List(Out)}, if for every valid input instance $\mathcal{O}$, \fl{u_1}, \ldots, \fl{u_k},
the set of outputs produced by the implementation is equal to
${\fl{set}(\fl{spec}}({\fl{sort}}_{\mathcal{O}}({\fl{u_1}}, \ldots, {\fl{u_k}})))$.
\end{definition}

\begin{theorem}[implementation correctness]
  \label{theorem:correctness} Any implementation produced by our framework is correct according to \Cref{def:distr-correctness}.
\end{theorem}

\section{Experimental Evaluation}
\label{sec:evaluation}

In this section we conduct a set of experiments to investigate tradeoffs between data parallelism and \emph{platform independence} in stream processing APIs.
That is, we want to distinguish between parallelism that is achieved automatically and parallelism that is achieved manually at the cost of portability when the details of the underlying platform change.
To frame this discussion,
we identify a set of platform independence principles (PIP) with which to evaluate this tradeoff:
\begin{description}
\item[PIP1:] \textbf{\emph{parallelism independence.}}
    Is the program developed without regard to the number of parallel instances, or does the program use the number of parallel instances in a nontrivial way?
\item[PIP2:] \textbf{\emph{partition independence.}}
    Is the program developed without regard to the correspondence between input data and parallel instances, or does it require knowledge of how input streams are partitioned to be correct?
\item[PIP3:] \textbf{\emph{API compliance.}}
    Does the program violate any assumptions made by the stream processing API?
\end{description}
Having identified these principles, the following questions guide our evaluation:
\begin{description}
\item[Q1]
For computations requiring synchronization, what are the throughput limits of automatic parallelism exposed by existing stream processing APIs?
\item[Q2]
Can \emph{manual} parallel implementations, i.e., implementations that may sacrifice \textbf{(PIP1--3)} above, that emulate synchronization plans
achieve \emph{absolute} throughput improvements in existing stream processing systems?
\item[Q3]
What is the throughput scalability of the synchronization plans that are generated automatically by our framework?
\item[Q4]
In summary, for each method of achieving data parallelism with synchronization, what platform independence tradeoffs are made?
\end{description}
In order to study these questions, we design three applications with synchronization requirements in \Cref{ssec:eval-applications}.
Our investigation compares three systems at different points
in the implementation space with varying APIs and performance characteristics.
First,
Apache Flink~\cite{Flink,carbone2015flink} represents a well-engineered mainstream streaming system with an expressive API.
Second, the Rust implementation of Timely Dataflow~\cite{Timely,murray2013naiad} (Timely) represents a system with support for iterative computation.
Third, \sys{} is a prototype implementation of our end-to-end framework that supports the communication patterns observed in arbitrary synchronization plans,
some of which are not supported by the execution models of Flink and Timely.

\heading{Experimental setup.}
We conduct all the experiments in this section in a distributed execution environment
consisting of AWS EC2 instances.
To account for the fact that AWS instances can introduce variability in results, we chose m6g medium (1 core @2.5GHz,
4~GB) instances, which do not use burst performance like free-tier instances.
We use instances in the same region (us-east-2) and we increase the number of instances for the scalability experiments.
Communication between nodes is managed by each respective system (the system runtime for Flink and Timely, and Erlang for \sys{}).

We configure Flink to be in true streaming mode by disabling batching (setting buffer-timeout to $0$), checkpointing, and dynamic adaptation.
For Timely, it is inherent to the computational model that events are batched by logical timestamp, and the system is not designed for event-by-event streaming, so our data generators follow this paradigm.
This results in significantly higher throughputs for Timely, but note that these throughputs are \emph{not} comparable with Flink and \sys{} due to the batching differences.
Because the purpose of our evaluation is not to compare absolute performance differences due to batching \emph{across systems},
we focus on relative speedups \emph{on the same system} and how they relate to platform independence \textbf{(PIP1--3)}.

\subsection{Applications requiring synchronization}
\label{ssec:eval-applications}

We consider three applications that require different forms of synchronization. All three of the applications do not perform CPU-heavy computation for each event so as to expose communication and underlying system costs in the measurements. The conclusions that we draw remain valid, since a computation-heavy application that would exhibit similar synchronization requirements would scale even better with the addition of more processing nodes.
The input for all three applications is synthetically generated.

\heading{Event-based windowing.}
An \emph{event-based window} is a window whose start and end is defined by the occurrence of certain events.
This results in a simple synchronization pattern where parallel nodes must synchronize at the end of each window.
For this application, we generate an input consisting of several streams of integer \emph{values} and a single (separate) stream of \emph{barriers}. The task is to
produce an aggregate of the values between every two consecutive
barriers, where \emph{between} is defined based on event timestamps.
We take the aggregation to be the sum of the values.
The computation is parallelizable if there are sufficiently more value events than barrier events. In the input to our experiments, there are 10K events in each value stream between two barriers.

\heading{Page-view join.}
The second application is an example of a streaming join.
The input contains 2 types
of events: \emph{page-view} events that represent visits of users to websites,
and \emph{update-page-info} events that update
the metadata of a particular website and also output the
old metadata when processed by the program. All of these
events contain a unique identifier identifying the website and the
goal is to join page-view events with the latest metadata of the
visited page to augment them for a later analysis. An additional
assumption is that the input is not uniformly distributed among
websites, but a small number of them receive most of the page-views.
To simulate this behavior in the inputs used in our experiments,
all views are distributed between two pages.

\heading{Fraud detection.}
Finally, the third application is a version of the
ML fraud detection application mentioned in the introduction, where the synchronization requirements are the same but the computation is simplified. The
input contains \emph{transaction} events and \emph{rule} events
both of which are integer values. On receiving a rule, the program
outputs an aggregate of the transactions since the last rule and a transaction is considered
fraudulent if it is equivalent modulo $1000$ to the sum of the previous aggregate (simulating model retraining)
and the last rule event. As in event-based windowing, we generate 10K transaction events between every two rule events.

\subsection{Implementations in Flink and Timely}
\label{ssec:eval-existing-implementations}

\begin{figure}[t]
  \centering
  \includegraphics[width=0.8\columnwidth]{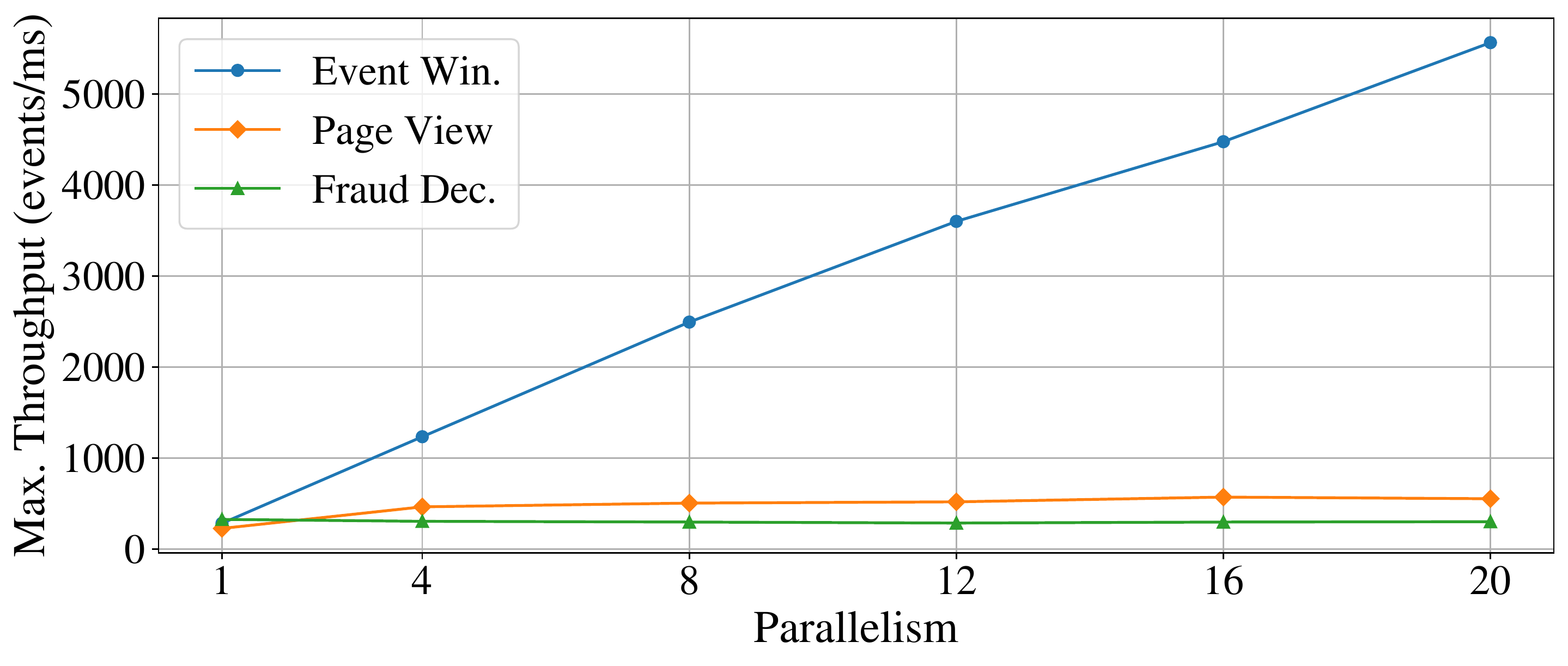}

  \includegraphics[width=0.8\columnwidth]{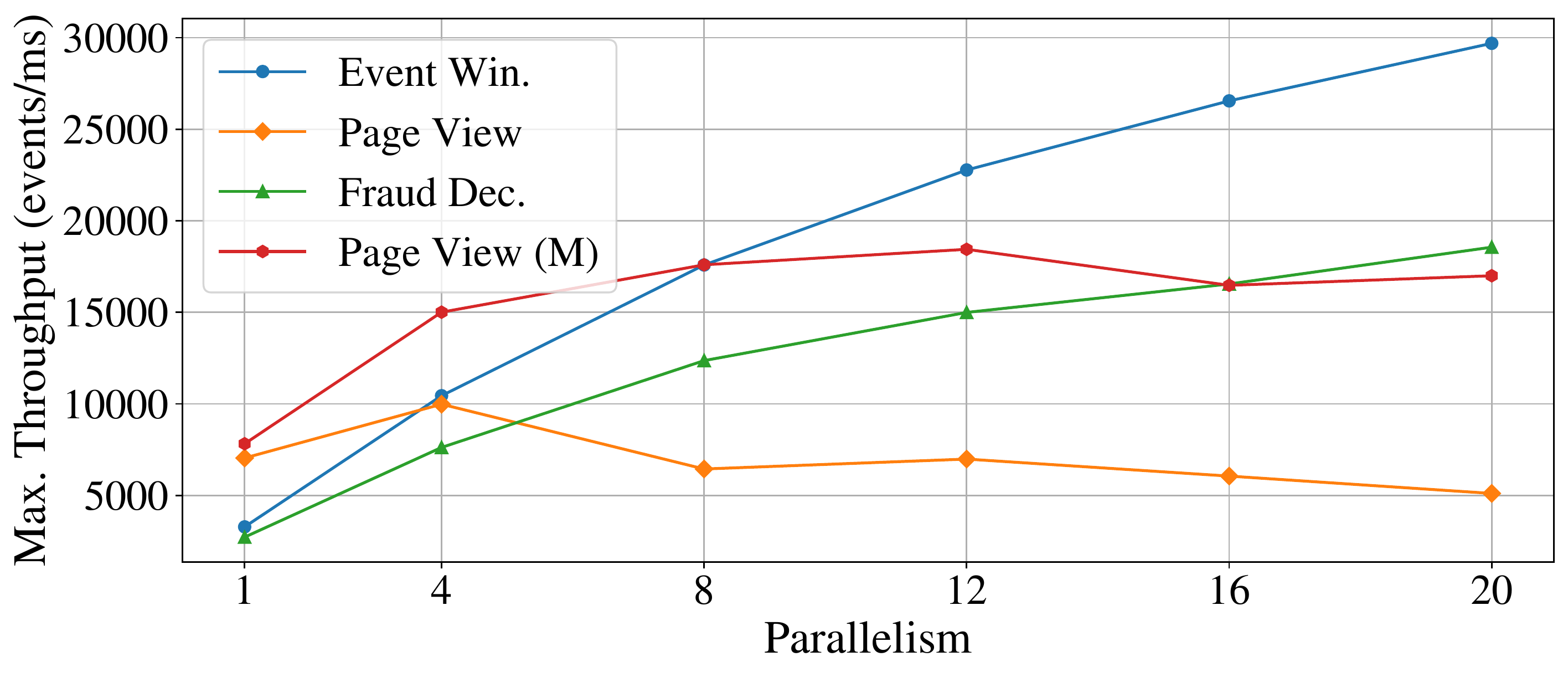}
  \caption{
  Flink (top) and Timely (bottom) maximum throughput increase with increasing parallel nodes for the three applications that require synchronization.
  For the page-view example in Timely, two implementations are shown: Page View uses automatic
  parallelism while Page View (M) is the manual parallel implementation
  in \Cref{fig:timely-snippet}.
  }
  \label{fig:existing-implementations-scaling}
\Description{Flink and Timely maximum throughput increase.}
\end{figure}

In this section we investigate how the Flink and Timely APIs can produce scalable parallel implementations for the aforementioned applications.
We iterated on different implementations resulting in a best-effort attempt to achieve good scalability. These implementations are summarized below, and the source code is presented in \Cref{appendix:timely-code,appendix:flink-code}.
For each of the implementations, we then ran an experiment where we increased the number of distributed nodes and measured the maximum throughput of the resulting implementation (by increasing the input rate until throughput stabilizes or the system crashes).
The results are shown in \Cref{fig:existing-implementations-scaling}.

\begin{figure}[t]
  \centering
  \small{}
\begin{verbatim}
updates.broadcast().filter(move |x| {
    x.name == page_partition_fun(NUM_PAGES, worker_index)
});
\end{verbatim}
  \caption{Snippet from the Timely manual (M) implementation of the page-view join example, satisfying \textbf{PIP1} and \textbf{PIP3} but not \textbf{PIP2}.
}
\Description{Snippet from the manual (M) Timely implementation of the page-view join example.}
\label{fig:timely-snippet}
\end{figure}

\heading{Event-based windowing.}
Flink's API supports a \ttt{broadcast} construct which can send barrier events to all parallel instances of the implementation, therefore being able to scale with an increasing parallel input load.
Note that Flink does not guarantee that the broadcast messages are synchronized with the other stream events, and therefore the user-written code has to ensure that it processes them in order.
By transforming these barriers to Flink watermarks that indicate window boundaries, we can then aggregate values of each window in parallel, and finally merge all windows to get the global aggregate.
Similarly, Timely includes a \ttt{broadcast} operator on streams, which sends all barrier events to all parallel instances;
then the \ttt{reclock} operator is used to match values with corresponding barriers and aggregate them.
Both the Flink and Timely implementations scale because the values are much more frequent than barriers, i.e., a barrier every 10K events, and are processed in a distributed manner.

\heading{Page-view join.}
The input of this application allows for data parallelism across keys, in this case websites, but also for the same key since some keys receive most of the events.
First, for both Flink and Timely, we implemented this application using a standard keyed join, ensuring that the resulting implementation will be parallel with respect to keys.
As shown by the scalability evaluation, this does not scale to beyond 4 nodes in the case that there are a small number of keys receiving most or all of the events.

We wanted to investigate whether it was possible to go beyond the automatic implementations and scale throughput for events \emph{of the same key}.
To study this, we provide a ``manual'' (M) implementation in Timely (\Cref{fig:timely-snippet} and \Cref{fig:existing-implementations-scaling}, bottom). Here, we broadcast \emph{update-page-info} events, then filter to only those relevant to each node, i.e. corresponding to \emph{page-view}s that it is responsible for processing. A similar implementation would be possible in Flink. Unfortunately, our implementation sacrifices \textbf{PIP2}, i.e., the assignment of events to parallel instances
becomes part of the application logic---there are explicit references to the physical partitioning of input streams
(\ttt{page\_partition\_fun}) and the the worker that processes each stream (\ttt{worker\_index}).
Additionally, the implementation broadcasts \emph{all} update events to \emph{all} sharded nodes (not just the ones that are responsible for them), introducing a linear synchronization overhead with the increase of the number of nodes.
An alternative choice would have been to not only broadcast events, but also keep state for \emph{every} page at every sharded replica: this would satisfy \textbf{PIP2} because nodes no longer need to be aware of which events they process, but it does not avoid the broadcasting issue and thus we would expect performance overheads.
Overall, we observe inability to automatically scale this application without sacrificing platform independence in both Flink and Timely.

\heading{Fraud detection.}
The standard dataflow streaming API cannot support cross-instance synchronization, and therefore we can only develop a sequential implementation of this application using Flink's API.
Timely offers a more expressive API with \emph{iterative} computation, and this allows for an automatically scaling implementation: after aggregating local state to globally updated the learning model, we have a cyclic loop which then sends the state back to all the nodes to process further events.
The results show that this implementation scales almost as well as the event-based window.
This effectively demonstrates the value of iterative computation in machine learning and complex stateful workflows.

\begin{takeaway}
\textbf{Take-away (Q1):}
The streaming APIs of Flink and Timely cannot automatically produce implementations that scale throughput for all applications that have synchronization requirements without sacrificing platform independence.
\end{takeaway}

\begin{figure}[t]
    \centering
    \begin{subfigure}[b]{0.48\columnwidth}
      \centering
      \includegraphics[width=\textwidth]{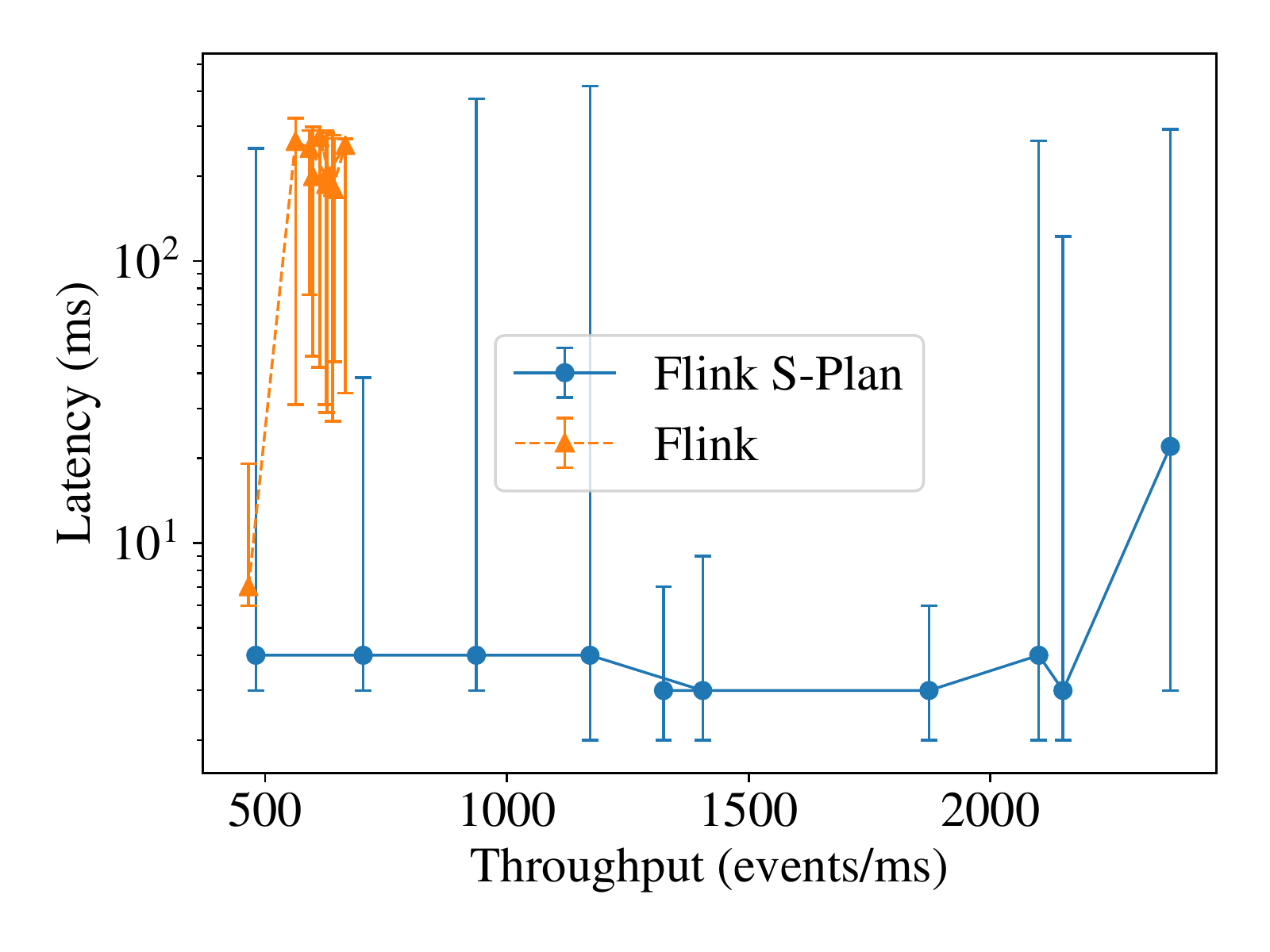}
      \caption{Page-view join.}
      \label{fig:synchronization-plan-page-view-join-throughput}
    \end{subfigure}
    ~
    \begin{subfigure}[b]{0.48\columnwidth}
      \centering
      \includegraphics[width=\textwidth]{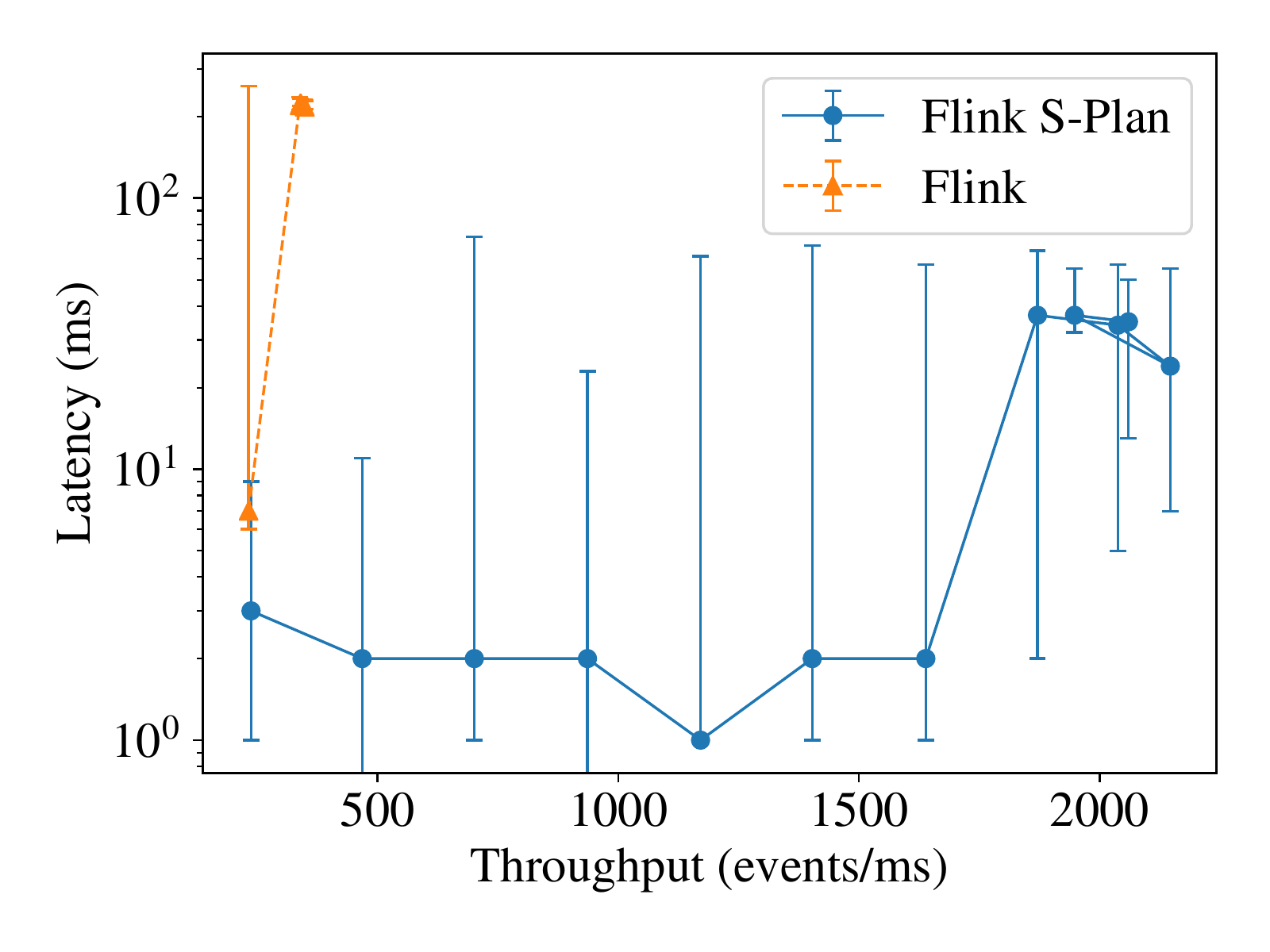}
      \caption{Fraud detection.}
      \label{fig:synchronization-plan-fraud-detection-throughput}
    \end{subfigure}
    \caption{
      Throughput (x-axis) and 10th, 50th, 90th percentile latencies on the y-axis for increasing input rates (from left to right) and 12 parallel nodes. Flink (\textcolor{Orange}{orange}) is the parallel implementation produced automatically, and Flink S-Plan (\textcolor{NavyBlue}{blue}) is the synchronization plan implementation.
      }
    \label{fig:synchronization-plan-throughputs-flink}
    \Description{Throughput scaling for Flink synchronization plan implementation.}
\end{figure}

\subsection{Manual synchronization}
\label{ssec:performance-evaluation}

\begin{figure}[t]
  \footnotesize{}
  \centering
  \begin{lstlisting}[language=Java,basicstyle=\scriptsize\ttfamily,commentstyle=\scriptsize\ttfamily,morekeywords={var}]
public Integer joinChild(
    final int subtaskIndex,
    final Integer state
) {
    final int parentId = subtaskIndex / pageViewParallelism;
    final int childId = subtaskIndex %
    joinSemaphores.get(parentId).get(childId).release();
    forkSemaphores.get(parentId).get(childId)
        .acquireUninterruptibly();
    return zipCode.get(parentId);
}
  \end{lstlisting}
  \caption{Snippet from the implementation of the manual synchronization \ttt{join} in Flink. This implementation does not satisfy \textbf{PIP1--3}.}
  \label{fig:flink-snippet}
  \Description{Snippet from the implementation of the manual synchronization join in Flink.}
\end{figure}

\noindent
To address Q2, we next investigate whether synchronization, implemented manually and possibly sacrificing \textbf{PIP1--3}, can offer concrete throughput speedups.
We focus this implementation in Flink, and consider the two applications that Flink cannot produce parallel implementations for, namely \emph{page-view join} and \emph{fraud detection}. We write a DGS program for these applications and we use our generation framework to produce a synchronization plan for a specific parallelism level (12 nodes). We then manually implement the communication pattern for these synchronization plans in Flink, and we measure their throughput and latency compared to the parallel implementations that the systems produced in \Cref{ssec:eval-existing-implementations}. The results for both applications are shown in \Cref{fig:synchronization-plan-throughputs-flink}. Flink does not achieve adequate parallelism and therefore cannot handle the increasing input rate (low throughput and high latency).

\heading{Page-view join.}
The synchronization plan that we implement for this application is a forest containing a tree for each key (website) and each of these trees has leaves that process page-view events. Each time an update event needs to be processed for a specific key, the responsible tree is joined, processes the event, and then forks the new state back.

\heading{Fraud detection.}
The synchronization plan that we implement for this application is a tree that processes rule events at its root and transactions at all of its leaves. The tree is joined in order to process rules and is then forked back to keep processing transactions.

\heading{Implementation in Flink.}
In order to implement the synchronization plans in Flink we need to introduce communication across parallel instances. We accomplish this by using a centralized service that can be accessed by the instances using Java RMI. Synchronization between a parent and its children happens using two sets of semaphores, $J$ and $F$.
A child releases its $J$ semaphore and acquires its $F$ semaphore when it is ready to join, and a parent acquires its children's $J$ semaphores, performs the event processing, and then releases their $F$ semaphores (\Cref{fig:flink-snippet}).
This implementation of manual synchronization
sacrifices all three platform independence principles \textbf{PIP1--3}.
For \textbf{PIP1} and \textbf{PIP2}, it refers explicitly to the number of parallel instances and the partitioning (\ttt{pageViewParallelism}, \ttt{subtaskIndex}, etc.).
For \textbf{PIP3}, it is not API-compliant because it uses an external service (semaphores) to implement synchronization, whereas Flink's documentation requires that operators lack side effects.
This requirement is imposed because, among other considerations, the use of semaphores might cause the program to fail in cases where
work is interrupted and/or repeated after a node failure.

\begin{takeaway}
\textbf{Take-away (Q2):}
Synchronization plans achieve higher throughputs (4-8x for 12 parallel nodes) over the automatic parallel implementations produced by Flink's API.
\end{takeaway}

\subsection{Implementation in \sys{}}
\label{ssec:eval-prototype-performance}

\begin{figure}[t]
  \centering
  \includegraphics[width=0.8\columnwidth]{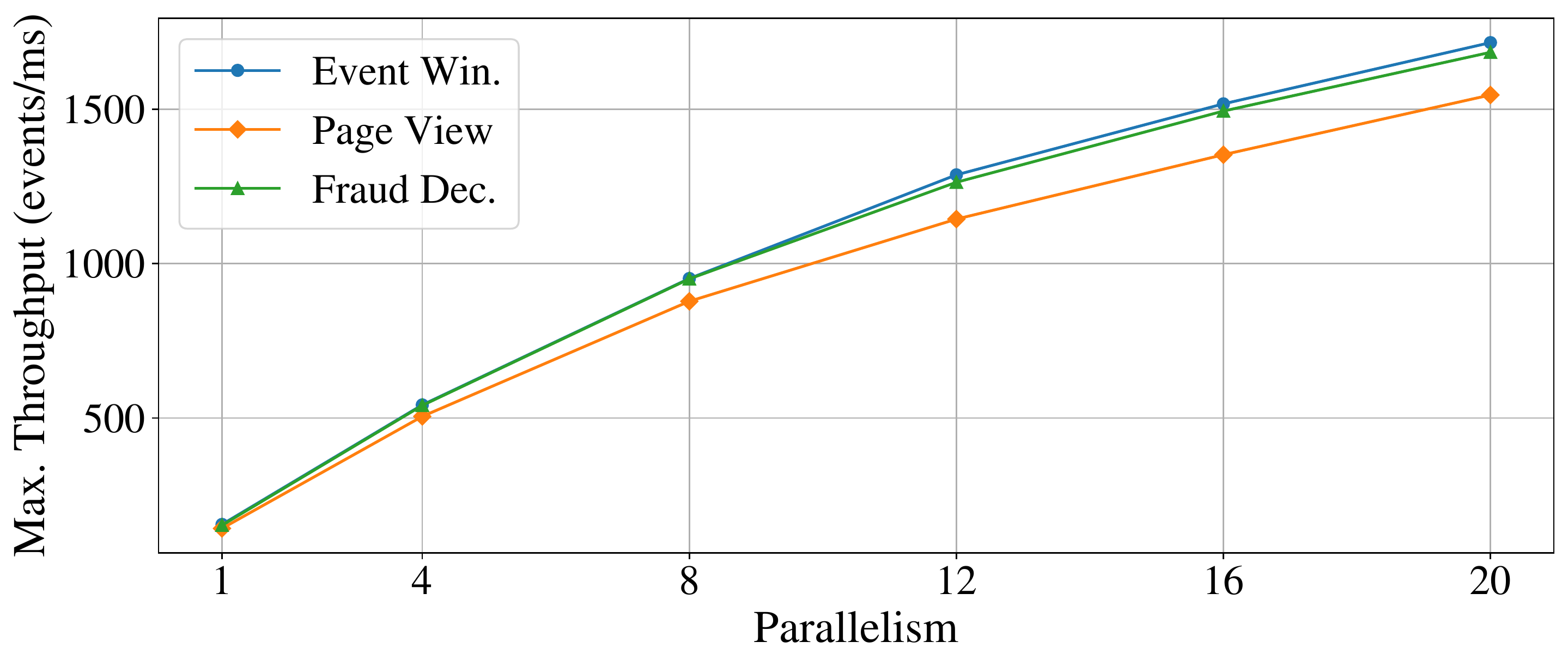}

  \caption{\sys{} (DGS) maximum throughput increase with increasing parallel nodes for the three applications.
  }
  \Description{Maximum throughput increase for our system.}
  \label{fig:flumina-scaling}
\end{figure}

To answer Q3, we implement \sys{},
a prototype of our end-to-end framework that can automatically achieve parallelism given a DGS program.
\sys{} receives a DGS program written in Erlang,
uses the generation framework that was described in \Cref{sec:dist-impl} to generate a correct and efficient synchronization plan,
and then implements the plan according to the description in \Cref{ssec:runtime}.
We implemented all three applications (source code described in \Cref{appendix:flumina-code})
in \sys{}, and measured maximum throughput increase with the addition of parallel nodes
(\Cref{fig:flumina-scaling}).

\begin{table*}[t]
  \normalsize
  \begin{center}
    \setlength{\tabcolsep}{4pt}
    \renewcommand{\arraystretch}{1.5}
    \begin{tabular}{l | b | t | f | b b | t t | f | b b | t | f}
		& \multicolumn{3}{c|}{Event window} & \multicolumn{5}{c|}{Page-view join} & \multicolumn{4}{c}{Fraud detection} \\
    \hline
    Development tradeoff
     & F & TD & DGS & F & FM & TD & TDM & DGS & F & FM & TD & DGS \\
    \hline
    \textbf{(PIP1)}
  Parallelism independence
		& \cmark & \cmark & \cmark
		& \cmark & \xmark & \cmark & \cmark & \cmark
		& \cmark & \xmark & \cmark & \cmark \\
	\textbf{(PIP2)}
  Partition independence
		& \cmark & \cmark & \cmark
		& \cmark & \xmark & \cmark & \xmark & \cmark
		& \cmark & \xmark & \cmark & \cmark \\
    \textbf{(PIP3)}
  API compliance
		& \cmark & \cmark & \cmark
		& \cmark & \xmark & \cmark & \cmark & \cmark
		& \cmark & \xmark & \cmark & \cmark \\
  Scaling
		& 10x & 8x & 8x
		& 2x & 9x & 1x & 2x & 8x
		& 1x & 9x & 6x & 8x \\
    \end{tabular}

    \phantom{LaTeX is incapable of changing the space in between figure and caption in a sane way.}
  \end{center}
\caption{
Development tradeoffs for each program,
together with throughput scaling for 12 nodes,
in Flink (F), Flink with manual synchronization (FM),
Timely (TD), Timely with manual partitioning (TDM), and our system (DGS).
}
\label{tab:developer-tradeoffs}
\end{table*}

\heading{Event-based windowing and fraud detection.}
The DGS program for event-based windowing contains: (i) a sequential update function that adds incoming value events to the state, and outputs the value of the state when processing a barrier event, (ii) a dependence relation that indicates that all events depend on barrier events, and (iii) a \fl{fork} that splits the current state in half, together with a \fl{join} that adds two states to aggregate them. The DGS program for fraud detection is the same with the addition that the  \fl{fork} also duplicates the sum of the previous transaction and last rule modulo 1000.

\heading{Page-view join.}
In addition to the sequential update function, the program indicates that events referring to different keys are independent, and that page-view events of the same key are also independent. The \fl{fork} and \fl{join} are very similar to the ones in \Cref{fig:key-value-store} and just separate the state with respect to the keys.

\begin{takeaway}
\textbf{Take-away (Q3):}
Across all three examples,
\sys{} produces parallel implementations that scale throughput without sacrificing platform independence.
\end{takeaway}

\subsection{Summary: development tradeoffs}
\label{ssec:eval-development-tradeoffs}

Finally, regarding Q4,
\Cref{tab:developer-tradeoffs} shows all the tradeoffs that need to be made for each of the programs in this section together with the throughput increase for 12 nodes.
Note that throughput scaling comparison is only relevant for the same system (each of which is denoted with a different color) and not across systems due to differing sequential baselines.
Of all the implementations in \Cref{ssec:eval-existing-implementations}, the Timely manual page-view example sacrifices \textbf{PIP2}.
The manually synchronizing Flink implementations in \Cref{ssec:performance-evaluation} show good throughput scaling at the cost of \textbf{PIP1--3}.

\begin{takeaway}
\textbf{Take-away (Q4):}
Of the three APIs studied,
only DGS can achieve scalable implementations across
all examples without sacrificing either parallelism independence, partition independence, or API compliance.
\end{takeaway}

\section{Related Work}
\label{sec:related-work}

\heading{Dataflow stream processing systems.}
Applications over streaming data can be implemented using
high-performance, fault tolerant stream processing systems, such as
Flink \cite{carbone2015flink}, Trill~\cite{chandramouli2014trill},
IBM Streams~\cite{HAG2013SPL},
Spark Streaming~\cite{DStreams2013}, Storm~\cite{Storm},
Samza~\cite{Samza2017}, Heron~\cite{kulkarni2015twitter-heron}, and
MillWheel~\cite{MillWheel}.
The need for synchronization in these systems has resulted in a number of extensions to their APIs, but they fall short of a general solution.
Naiad~\cite{murray2013naiad} proposes \emph{timely dataflow} in order
to support iterative computation, which enables some synchronization but falls short of automatically scaling without high-level design sacrifices, as shown in our evaluation.
S-Store and TSpoon~\cite{meehan2015s,affetti2020tspoon} extend stream processing systems with online transaction processing (OLTP),
which includes some forms of synchronization, e.g. locking-based concurrency control.
Concurrent with our work, Nova~\cite{zhao2021timestamped} also identifies
the need for synchronization in stream processing systems,
and proposes to address it through a shared state abstraction.

\heading{Actor-based databases.}
As data processing applications are becoming more complex, evolving from
data analytics to general event-driven applications, some stream
processing and database systems are moving from dataflow programming to more
general actor models
\cite{DBLP:conf/sigmod/CarboneFKK20,DBLP:conf/cidr/Bernstein19,DBLP:conf/cidr/BernsteinDKM17,DBLP:conf/sigmod/0001S18,xu2021move}.
For example, Flink has recently released Stateful Functions,
an actor-based programming model running on top of Flink
\cite{DBLP:journals/pvldb/AkhterFK19,StatefulFunctions}.
Actor models can encode arbitrary synchronization patterns,
but the patterns still need to be implemented manually as
message-passing protocols.
DGS and synchronization plans can be built on top
of the actor abstraction, and in fact our own implementation
relies on actors as provided by Erlang~\cite{armstrong1993erlang}.

\heading{Programming with synchronization.}
In the broader context of distributed and parallel programming,
synchronization is a significant source of overhead for developers,
and a good deal of existing work can be viewed as addressing this problem.
Our model draws inspiration from fork-join based
concurrent programming~\cite{frigo1998implementation,lea2000java},
  bringing some of the expressiveness in those models to the streaming setting,
  where parallelism is much less flexible but essential for performance,
  but also extending them, since in our setting the system (and not the user) decides when to fork and join by choosing a synchronization plan.
A particularly relevant example is
Concurrent Revisions~\cite{burckhardt2010concurrent},
which is a programming model that guarantees determinism in the presence of concurrent updates by allowing programmers to specify isolation types that are processed in parallel and then merged at join points.
The difference of our work is that it targets a more restricted domain providing automation,
not requiring programmers to manually specify the execution synchronization points.
Another related domain is monotonic lattice-based programming models,
including
Conflict-Free Replicated Data Types~\cite{shapiro2011conflict},
Bloom$^L$~\cite{conway12},
and LVars~\cite{lvars13,lvars14},
which are designed for coordination-free distributed programming.
These models guarantee strong eventual consistency,
i.e., eventually all replicas will have the same state,
but, in contrast to our model, CRDTs and Bloom$^L$
do not allow synchronization between different workers.
LVars, which focuses on determinism for concurrent updates on shared variables,
extends lattice-based models with a freeze operation that enforces a synchronization point,
inducing partial order executions that are similar to the ones in our model.
Some similarities with our work can be found in the domain of consistency for replicated data stores.
Some examples include RedBlue consistency~\cite{li2012making},
MixT~\cite{milano2018mixt},
Quelea~\cite{sivaramakrishnan2015declarative},
CISE~\cite{gotsman16},
Carol~\cite{lewchenko2019sequential},
all of which support a mix of consistency guarantees on different operations,
inducing a partial order of data store operations.

\heading{Correctness in stream processing.}
Finally, in prior work,
researchers have pointed out that
parallelization in stream processing systems is not
semantics-preserving, and have proposed methods to restrict
parallelization so that it preserves the
semantics~\cite{schneider2015safeparallelism, mamouras2019data}.
In particular, dependence relations have been previously used in this context
as to specify and ensure correct parallelism~\cite{mamouras2019data,2020:DifferentialTesting:OOPSLA}
and as a type system for synchronization~\cite{alur2021synchronization}.
However, these works do not propose a general programming model
or generation of a parallel implementation.

\section{Future Work}
\label{sec:conclusion}

One problem that is not yet adequately explored in our framework is optimization:
given a DGS program, how to select a valid synchronization plan which is most efficient according to a desired cost metric.
Traditional optimization algorithms for stream processing systems cannot be directly applied to
the complex tree structure of synchronization plans.
There are also possibilities for
dynamic optimization, in which the synchronization plan is modified online in response to profiling data from the system.
Besides optimization, the implementation of synchronization plans needs to address a plethora of systems related issues, such as (i) the efficient management and communication of forked state in a distributed environment, (ii) execution guarantees in the presence of faults, and (iii) supporting performance optimizations such as batching and backpressure.

\begin{acks}
  We want to thank
    Anton Xue,
    Badrish Chandramouli,
    Nikos Vasilakis,
    and many others,
    including attendees at the
    Heidelberg Laureate Forum 2019 and
    PRECISE Industry Day 2019,
    for their comments on earlier versions of this work over the years.
  We also want to thank the reviewers of this and all past versions of this paper for their valuable feedback leading to substantial improvements.
  This research was supported in part by NSF award 1763514.
\end{acks}

\bibliographystyle{ACM-Reference-Format}
\bibliography{ref}

\newpage
\appendix
\section{Case Studies}
\label{appendix:case-studies}

In this section, we evaluate if DGS can be used for realistic application workloads.
We consider the following questions:
\begin{enumerate}
    \item How does the performance of \sys{} compare with handcrafted implementations?
    \item What is the additional programming effort necessary to achieve automatic parallelization when using the DGS programming model?
\end{enumerate}

To evaluate these questions we describe two case studies on applications taken from the literature that have existing high-performance implementations for comparison. The first is a state-of-the-art algorithm for statistical outlier detection, and the second is a smart home power prediction task from the DEBS Grand Challenge competition.
For question (1),
the two case studies are posed in the literature targetting different performance metrics: throughput scalability for outlier detection and latency for the smart home task.
We are able to achieve performance comparable to the existing handcrafted implementations in both cases with respect to the targetted metric.
For question (2), this performance is achieved while putting minimal additional effort into parallelization: compared to 200-700 LoC for the sequential task, writing the fork and join primitives requires only an additional 50-60 LoC.
These results support the feasibility of using DGS for practical workloads, with minimal additional programmer effort to accomplish parallelism.

\subsection{Statistical outlier detection}
\label{ssec:outlier-detection}

\textsc{Reloaded}~\cite{otey2006fast} is a state-of-the-art distributed
streaming algorithm for outlier detection in mixed attribute
datasets. It is a structurally similar
to the fraud-detection example from the experimental evaluation.
The algorithm assumes a set of input streams that
contain events drawn from the same distribution. Each input stream is
processed independently by a different worker with the goal of
identifying outlier events. Each worker constructs a local model of
the input distribution and uses that to flag some events as potential
outliers. Whenever a user (or some other event) requests the current
outliers, the individual workers merge their states to construct a
global model of the input distribution and use that to flag the
potential outlier events as definitive outliers.
We executed a network
intrusion detection task from the original
paper~\cite{kddcup1999dataset}.
The goal of the task is to distinguish malicious connections from a
streaming dataset of simulated connections on a typical U.S.~Air Force
LAN. Each connection is represented as an input event. In the
experiment we varied the number of nodes from 1-8 and we measured the
execution time speedup. We executed this experiment on a local server
(Intel Xeon Gold 6154, 18 cores @3GHz, 384 GB) with the NS3~\cite{carneiro2010ns3} network simulator
to have comparable network latency with the original paper that used a small local cluster.

\heading{Programmability.}
The sequential implementation of the algorithm in DGS consists of approximately 700
lines of Erlang code---most of which is boilerplate to manage the
data structures of the algorithm.
Compared to the sequential implementation, the programming effort to achieve a
parallel implementation is a straightforward pair of \texttt{fork} and \texttt{join} primitives, which amounts to 50~LoC.

\heading{Performance.}
DGS and synchronization plans can capture the complex synchronization pattern that was proposed for \textsc{Reloaded}, achieving comparable speedup to that reported in the original paper:
almost linear (7.3$\times$ for 8 nodes), compared to 7.7$\times$ for 8 nodes by the handcrafted C++ cluster evaluation reported in the paper.

\subsection{IoT power prediction}
\label{ssec:iot-case-study}

The DEBS Grand Challenge is an annual competition to evaluate how
research solutions in distributed event-based systems can perform on
real applications and workloads.  The 2014
challenge~\cite{jerzak2014debs2014, DEBS2014web} involves an IoT task
of power prediction for smart home plugs.  As with the previous case
study, our goal is to see if the performance and programmability of
our model and framework can be used on a task where there are
state-of-the-art results we can compare to.
The problem (query 1 of the challenge) is to predict the load of a
power system in multiple granularities (per plug, per house, per
household) using a combination of real-time and historic data.  We
developed a solution that follows the suggested prediction method,
that involves using a weighted combination of averages from the hour and
historic average loads by the time of day.  This task involves
inherent synchronization: while parallelization is possible at each of
the different granularities, synchronization is then required to bring
together the historic data for future predictions. For example, if
state is parallelized by plug, then state needs to be joined in order
to calculate a historic average load by household.  Our program
is conceptually similar to the map from keys to counters in
\Cref{sec:prog-model}, where we maintain a map of historical totals
for various keys (plugs, houses, and households).  The challenge input
contains 29GB of synthetic load measurements for 2125 plugs
distributed across 40 houses, during the period of one month.  We
executed our implementation on a subset of 20 of the 40 houses, which
we sped up by a factor of 360 so that the whole experiment took about
2 hours to complete.  To compare with submissions to the challenge
which were evaluated on one node, we ran a parallelized computation on
one server node.
To simulate the network and measure network load, we
then used NS3 \cite{carneiro2010ns3}.

In the state, we maintain a set of maps of recent and historical
averages for house, household, and plug.  Then, we set different
houses, $\tg{house}_k$ (for $k$ between $1$ and $20$) to be
different tags, and we add an end-timeslice event $\tg{\#}$
at the end of every hour.  The dependence relation is that
end-timeslice is dependent on everything (this is when output is
produced), and that $\tg{house}_k$ is dependent on itself for every
$k$, but independent of other houses.  The \emph{fork} function
splits each map by house ID, and the \emph{join} function merges
maps back together.

\heading{Programmability.}
In total, the sequential code of our solution amounted to about 200~LoC, and the parallelization primitives (fork, join,
dependence relation) were 60~LoC.  We conclude that the overhead to enable
parallelization is small compared to the sequential code.

\heading{Performance.}
Latency varied between $44$ms (10th percentile), $51$ms (median), and
$75$ms (90th), and the average throughput was $104$
events/ms. These results are on par with the ones reported
by that year's grand challenge winner~\cite{mutschler2014predictive}:
$6.9ms$ (10th) $22.5ms$ (median) 41.3 (90th) and $131$ events/ms
throughput.
Note that although the dataset is the same,
the power prediction method used was different in some solutions.  In
this application domain, our optimizer has the advantage of enabling
edge processing: we measure only 362~MB of data sent over the network,
in contrast to the 29~GB of total processed data.

\section{Communication Optimizer}
\label{appendix:optimizer}

In this section we describe one of the synchronization plan optimizers that we have developed, namely one that is based on a simple
heuristic: it tries to generate a synchronization plan with a
separate worker for each input stream, and then tries to place these
workers in a way that minimizes communication between them.
This optimizer assumes a network of computer nodes and takes as input
estimates of the input rates at each computer node. It searches for an
$P$-valid synchronization plan that maximizes the number of events
that are processed by leaves; since leaves can process events
independently without blocking. The optimizer first uses a greedy
algorithm that generates a graph of implementation tags (where the
edges are between dependent tags) and iteratively removes the implementation
tags with the lowest rate until it ends up with at least two
disconnected components.

\begin{example}
\label{example:optimization}
For an example optimization run, consider \Cref{fig:example-configuration},
and suppose that $\tg{r}(2)$ has a rate of 10 and arrives at node $E_0$,
$\tg{r}(1)$, $\tg{i}(1)$ have rates of 15 and 100 respectively and arrive at node $E_1$,
$\tg{i}(2)_a$ has rate 200 and arrives at node $E_2$,
and $\tg{i}(2)_b$ has rate 300 and arrives at node $E_3$.
Since events of different keys are independent, there are
two connected components in the initial graph---one for each
key. The optimizer starts by separating them into two
subtrees. It then recurses on each disconnected component, until there
is no implementation tag left, ending up with the tree structure shown
in \Cref{fig:example-configuration}. Finally, the optimizer
exhaustively tries to match this implementation tag tree, with a
sequence of forks, in order to produce a valid synchronization plan
annotated with state types, updates, forks, and joins.
This generates the implementation in \Cref{fig:distr_arch}.
\end{example}

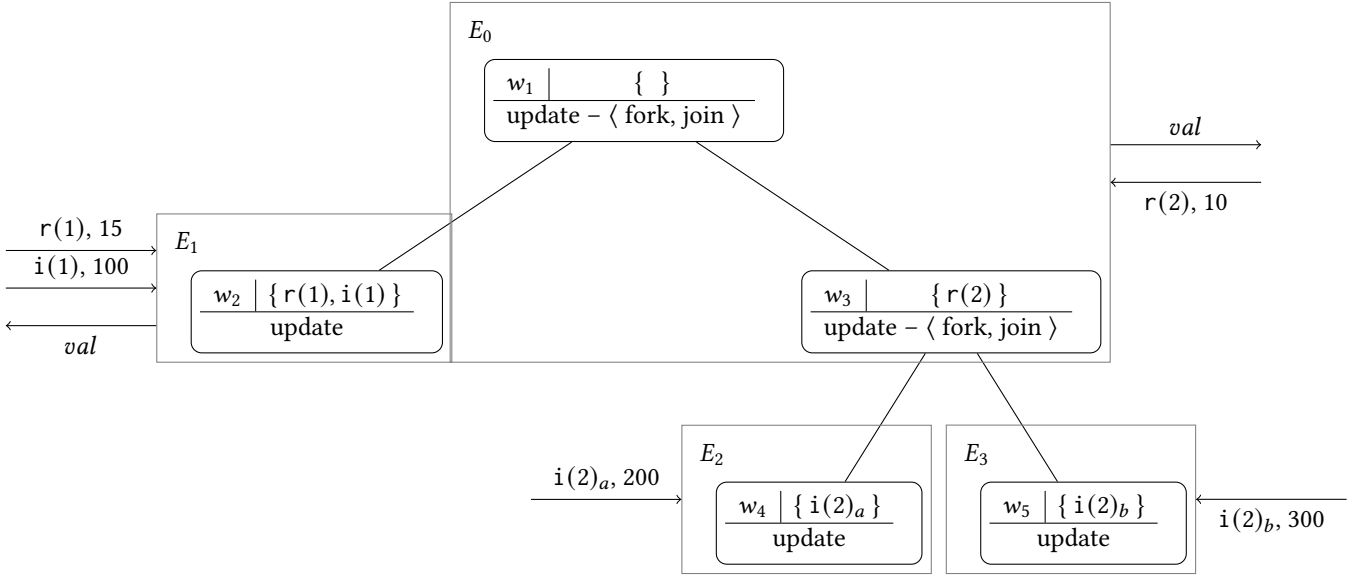
\begin{figure*}[t]
  \begin{tikzpicture}[sibling distance=25em, level distance=80pt,
    every node/.style = {shape=rectangle,
      rounded corners,
      draw, align=center},
    level 1/.style = {sibling distance=24em},
    level 2/.style = {sibling distance=10em}]]

    \node (e0) { \TopDNode{$w_1$}{ }
                          {$E_0$}{update -- $\langle$ fork, join $\rangle$} }
      child { \DNode{$w_2$}{$\tg{r}(1), \tg{i}(1)$}{$E_1$}{update}{e1} }
      child { \DNode{$w_3$}{$\tg{r}(2)$}{$E_0$}{update -- $\langle$ fork, join $\rangle$}{e23}
          child { \DNode{$w_4$}{$\tg{i}(2)_a$}{$E_2$}{update}{e2} }
          child { \DNode{$w_5$}{$\tg{i}(2)_b$}{$E_3$}{update}{e3} } };

      \node (node0) [above=1mm of e0.north west, draw=none] { $E_0$ };
      \node [draw=black!50, fit={(e0) (e23) (node0)}, rounded corners=0] (c0) {};

      \node (node1) [above=1mm of e1.north west, draw=none] { $E_1$ };
      \node [draw=black!50, fit={(e1) (node1)}, rounded corners=0] (c1) {};

      \node (node2) [above=1mm of e2.north west, draw=none] { $E_2$ };
      \node [draw=black!50, fit={(e2) (node2)}, rounded corners=0] (c2) {};

      \node (node3) [above=1mm of e3.north west, draw=none] { $E_3$ };
      \node [draw=black!50, fit={(e3) (node3)}, rounded corners=0] (c3) {};

      \coordinate[right=20mm of c0] (d0);
      \coordinate[above=5mm of d0] (dd0);
      \coordinate[above=5mm of c0.east] (dc0);
      \draw [->] (d0) to[left] node[draw=none, auto] {$\tg{r}(2)$, 10} (c0);
      \draw [->] (dc0) to[right] node[draw=none, auto] {$val$} (dd0);

      \coordinate[left=20mm of c1] (d1);
      \coordinate[above=5mm of d1] (dd1);
      \coordinate[above=5mm of c1.west] (dc1);
      \coordinate[below=5mm of d1] (dd2);
      \coordinate[below=5mm of c1.west] (dc2);
      \draw [->] (d1) to[right] node[draw=none, auto] {$\tg{i}(1)$, 100} (c1);
      \draw [->] (dd1) to[right] node[draw=none, auto] {$\tg{r}(1)$, 15} (dc1);
      \draw [->] (dc2) to[left] node[draw=none, auto] {$val$} (dd2);

      \coordinate[left=20mm of c2] (d2);
      \draw [->] (d2) to[left] node[draw=none, auto] {$\tg{i}(2)_a$, 200} (c2);

      \coordinate[right=20mm of c3] (d3);
      \draw [->] (d3) to[right] node[draw=none, auto] {$\tg{i}(2)_b$, 300} (c3);

  \end{tikzpicture}
  \caption{Example synchronization plan generated in \Cref{example:optimization}.
  The large gray rectangles $E_0$, $E_1$, $E_2$, $E_3$ represent physical nodes
  and the incoming arrows represent input streams and their relative rates.}
  \label{fig:distr_arch}
\end{figure*}

\section{Correctness Proof}
\label{appendix:correctness}

We prove \Cref{theorem:correctness} and therefore show that our framework is correct according to \Cref{def:distr-correctness}.
The key assumptions used are:
\begin{enumerate}
\item[(1)] The program is consistent, as defined in \Cref{ssec:prog-model-correctness}.
\item[(2)] The input streams constitute a valid input instance, as defined in \Cref{def:valid-input-instance}.
\item[(3)] The synchronization plan that is chosen by the optimizer is valid, as defined in \Cref{ssec:distributed-configurations}.
\item[(4)] Messages that are sent between two processes in the system arrive in order and are always received. This last assumption is ensured by the Erlang runtime.
\item[(5)] The underlying scheduler is fair, i.e. all processes get scheduled infinitely often.
\end{enumerate}

The proof decomposes the implementation into the mailbox
implementations and the worker implementations, which are both sets of
processes that given a set of input streams produce a set of output
streams. We first show that the mailboxes transform a valid input
instance to a worker input that is correct up to reordering of
independent events. Then we show that given a valid worker input, the
workers processes' collective output is correct.
The second step relies on
\Cref{thm:consistency-implies-determinism},
from the previous section, as well as
\Cref{lemma:worker-wire-correspondence}
which ties this to the mailbox implementations,
showing that the implementation produces a valid wire diagram according to the formal semantics of the program $P$.
Given a valid input instance $u$ and a computation program $P$ that contains
in particular a sequential specification \fl{spec: List(Event) -> List(Out)},
we want to show that any implementation $f$ that is produced by our framework is correct according to \Cref{def:distr-correctness}.

\begin{definition}
For a valid input instance $u$, a worker input $m_u = (m_1, ..., m_N)$
is \emph{valid} with respect to $u$ if $m_i \in \mathrm{reorder}_{D_I}
(\mathrm{filter}(\mathrm{rec}_{w_i}, \mathrm{sort}_O(u))),$ where
$\mathrm{reorder}_{D_I}$ is the set of reorderings that preserve the
order of dependent events,
  $\mathrm{rec}_w(\langle \sigma, t, p\rangle) = \sigma \in \mathrm{atags}(w) \cup \wfield{w}{itags}$
  is a predicate of the messages that each worker $w$ receives,
  and $\mathrm{atags}$ is the set of implementation tags that are handled by a workers ancestors, and is given by $\mathrm{atags}_w = \{ \wfield{w'}{itags} : \forall w' \in \mathrm{anc}(w', w)\}$.
\end{definition}

\begin{lemma}
\label{lemma:mailbox}
Given a valid input instance $u$, any worker input $m$ produced by
the mailbox implementations $(u, m)$ in $f$ is valid with respect to
$u$.
\end{lemma}
\begin{proof}
By induction on the input events of one mailbox and using assumptions
(2)-(5).
\end{proof}

Each worker $w$ runs the update function on each event $e = \langle
\sigma, t, p\rangle$ on its stream that it is responsible for $\sigma
\in \wfield{w}{itags}$ possibly producing output $\fl{o: List(Out)}$.
For all events in the stream that one of its
ancestors is responsible for, it sends its state to its parent worker
and waits until it receives an updated one.
The following key lemma states that this corresponds to a particular
wire diagram in the semantics of the program.

\begin{lemma}
\label{lemma:worker-wire-correspondence}
Let $m$ be the worker input to $f$ for program $P$
on input \fl{u},
and let $\fl{o}_i: \fl{List(Out)}$ be the stream of output events
produced by worker $i$ on input $m_i$.
Then there exists \fl{v: List(Out)}
such that $\fl{inter}(\fl{v}, \fl{o}_1, \fl{o}_2, \ldots, \fl{o}_N)$
and $(\fl{u}, \fl{v}) \in \sem{S}$.
\end{lemma}
\begin{proof}
By induction on the worker input and using assumption (3), in particular
validity condition (V1),
we first show that the worker input corresponds to a wire diagram,
in particular we show that
$\semantics{\fl{State_0}}{\fl{true}}{\fl{s}}{\fl{u}'}{\fl{s'}}{\fl{v}}$
where \fl{v} is an interleaving of $\fl{o}_1, \ldots, \fl{o}_N$
and $\fl{u}'$ is \emph{any} interleaving of the events
$\fl{u}'_i$
processed by each mailbox, namely
$\mathrm{filter}(\mathrm{rec}_{w_i}, \wfield{w_i}{itags})$.
Applying \Cref{lemma:mailbox},
\fl{u} is one possible interleaving of the events $\fl{u}'_i$
and hence we conclude that
$\semantics{\fl{State_0}}{\fl{true}}{\fl{s}}{\fl{u}}{\fl{s'}}{\fl{v}}$,
thus
$(\fl{u}, \fl{v}) \in \sem{S}$.
\end{proof}

Combining \Cref{lemma:worker-wire-correspondence}
and \Cref{thm:consistency-implies-determinism} then yields the end-to-end correctness theorem \Cref{theorem:correctness}.

\section{\sys{} System Details}
\label{appendix:flumina-impl}

\subsection{\sys{} synchronization latency}
\label{ssec:latency-eval}

We also conducted an experiment to evaluate the latency of synchronization plans. We studied three factors that affect latency: (i) the depth of the synchronization plan, (ii) the rate of events that are processed at non-leaf nodes of the plan, and (iii) the heartbeat rate. We ran the event-based windowing application with various configurations and the results are shown in \Cref{fig:synchronization-overhead}.
\begin{figure}[t]
    \centering
    \begin{subfigure}[t]{0.45\columnwidth}
      \includegraphics[width=\textwidth]{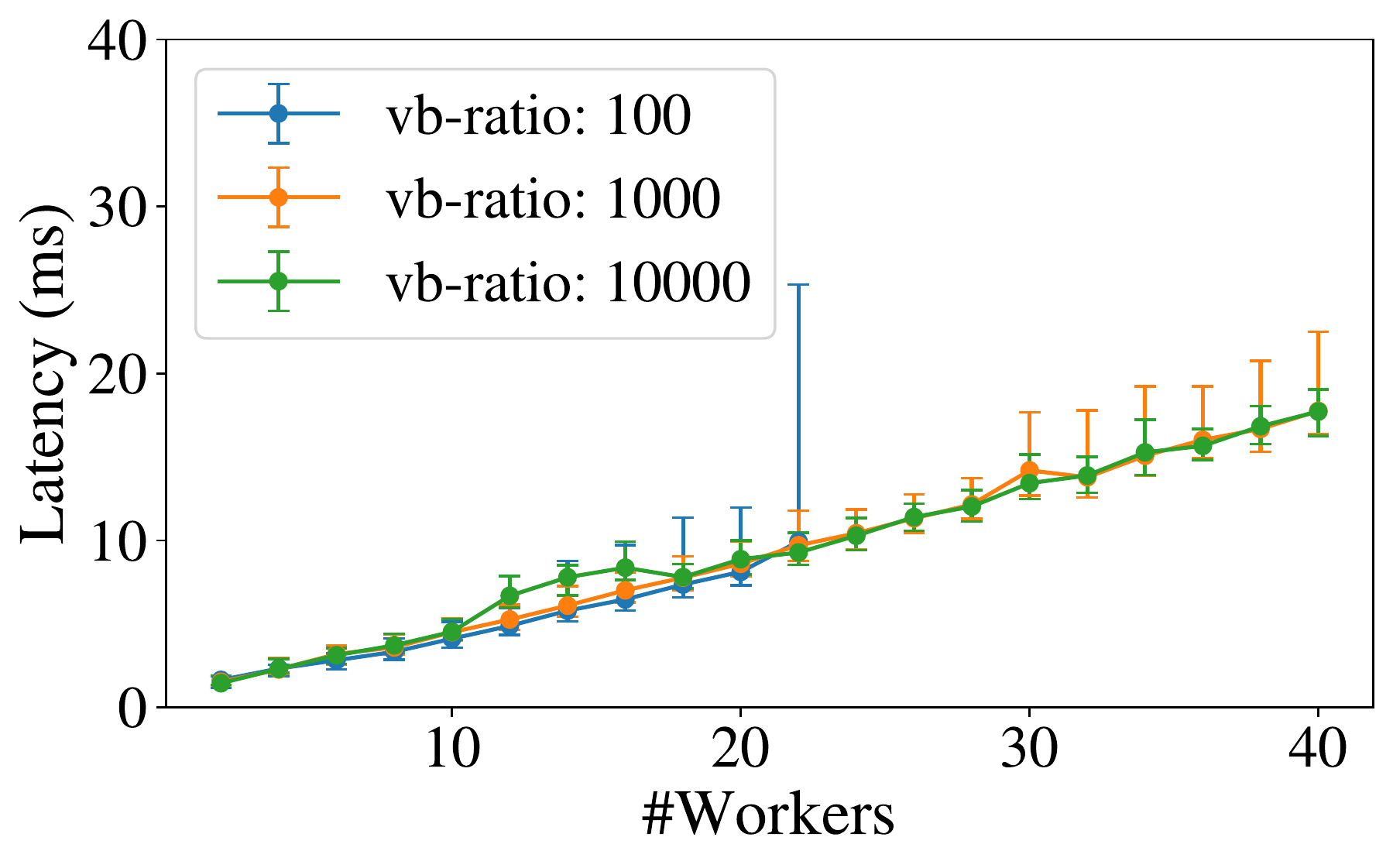}
    \end{subfigure}
    \begin{subfigure}[t]{0.45\columnwidth}
      \includegraphics[width=\textwidth]{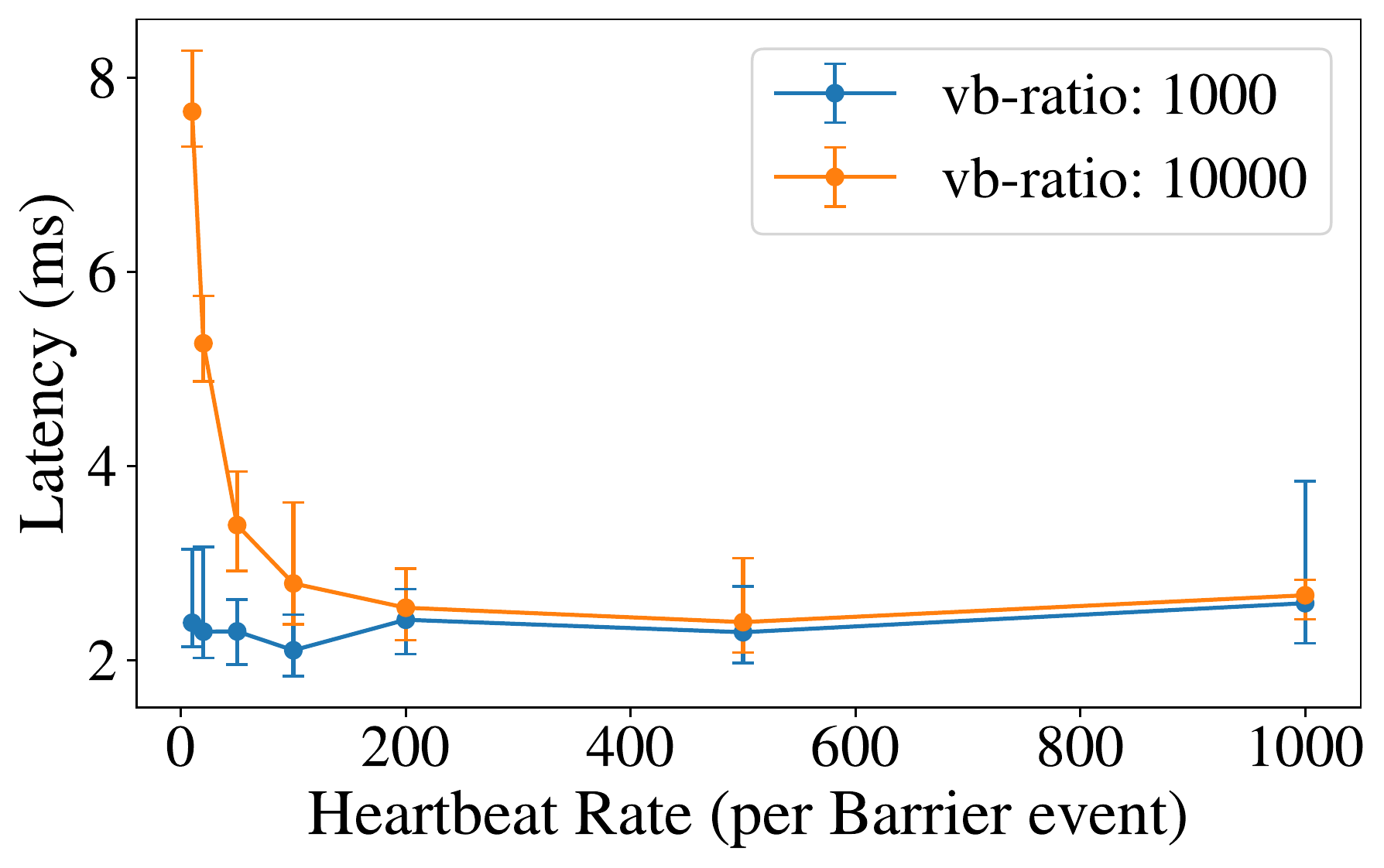}
    \end{subfigure}
    \caption{
    \sys{} latency (10/50/90 percentile) on the event-based windowing application for various configurations:
    Synchronization impact on latency (10/50/90 percentile) for various configurations:
    \textbf{(a)}
    Varying ratios of value to barrier events (different lines) and number of parallel nodes (x-axis). Heartbeat ratio is 1/100 the vb-ratio. The blue line stops after 22 workers.
    \textbf{(b)} Varying ratios of value to barrier events (different lines) and heartbeat rates (x-axis). Number of parallel nodes is fixed to 5.
      }
    \label{fig:synchronization-overhead}
\end{figure}
Latency increases linearly with
the number of workers, since the more workers there are, the more
messages have to be exchanged when a barrier event occurs. Note that
the latencies are higher for lower vb-ratios since every time a barrier
event occurs, all of the nodes need to synchronize. In particular, the
system cannot handle beyond 22 workers for vb-ratio of 100 (i.e. 100
forks-joins of the whole synchronization plan per second). On the right,
we see that if the heartbeat rate is too low latency increases since worker mailboxes cannot release events quickly and therefore get filled up, only releasing events in big batches whenever a barrier occurs.

\subsection{State checkpoints}

\sys{} provides a simple yet effective state checkpointing mechanism
for resuming a computation in case of faults.
In contrast to other distributed stream processing systems, where the
main challenge for checkpointing state is the acquisition of
consistent distributed
snapshots~\cite{murray2013naiad,carbone2017statemanagementflink},
performing a state checkpoint is straightforward when the root
node has joined the states of its descendants.
The joined state at that point in time (which corresponds to
the timestamp of the message that triggered the join request) is a
consistent snapshot of the distributed state of the system, and
therefore we get a non pausing snapshot acquiring mechanism for free.
We have exploited this property of \sys{} to implement a programmable
checkpoint mechanism that is given to the system as an option when
initializing worker processes. The checkpoint mechanism can be
instantiated to save a snapshot of the state every time the state is
joined on the root node, or less frequently, depending on an arbitrary
user defined predicate. We implemented a checkpoint mechanism that
produces a checkpoint every time the root node joins its children
states.

\section{\sys{} Programs}
\label{appendix:flumina-code}

This section contains the real code for the implementations of the three programs in \Cref{sec:evaluation}. Note that the code for fraud detection (\Cref{fig:fraud-detection-erlang}) is very similar to the code for value barrier (\Cref{fig:value-barrier-erlang}) with the most important difference being that fork returns a sum that was computed from the previous values down the tree.

\begin{figure}[H]
  \centering
  \footnotesize{}
\begin{FluminaCode}
-type ab_tag() :: {'a', integer()} | 'b'.
-type event() :: {ab_tag(), integer()}.
-type state() :: integer().

depends(ATags) ->
    ADeps = [{ATag, [b]} || ATag <- ATags],
    BDeps = {b, ATags},
    maps:from_list([BDeps|ADeps]).

-spec init() -> state().
init() -> 0.

-spec update(event(), state(), mailbox()) -> state().
update({{a,_}, Value}, Sum, SendTo) ->
    Sum + Value;
update({b, Ts}, Sum, SendTo) ->
    SendTo ! {sum, {{b, Ts}, Sum}},
    Sum.

-spec join(state(), state()) -> state().
join(Sum1, Sum2) ->
    Sum1 + Sum2.

-spec fork({tag_predicate(), tag_predicate()}, state())
        -> {state(), state()}.
fork(_, Sum) ->
    {Sum, 0}.

\end{FluminaCode}
\caption{\sys{} implementation of the Value Barrier example from \Cref{sec:evaluation}.}
\label{fig:value-barrier-erlang}
\end{figure}

\begin{figure*}[tbhp]
  \centering
  \footnotesize{}
\begin{FluminaCode}
-type uid() :: integer().
-type uids() :: list(uid()).
-type zipcode() :: integer().
-type update_user_address_tag() :: {'update_user_address', uid()}.
-type get_user_address_tag() :: {'get_user_address', uid()}.
-type page_view_tag() :: {'page_view', uid()}.
-type event_tag() :: update_user_address_tag() | get_user_address_tag() | page_view_tag().
-type update_user_address() :: {update_user_address_tag(), {zipcode(), integer()}}.
-type get_user_address() :: {get_user_address_tag(), integer()}.
-type page_view() :: {page_view_tag(), integer()}.
-type event() :: update_user_address() | get_user_address() | page_view().

-spec depends(uids()) -> depends().
depends(Keys) ->
    PageViewDeps = [{{page_view, Key}, [{update_user_address, Key}]} || Key <- Keys],
    GetAddressDeps = [{{get_user_address, Key}, [{update_user_address, Key}]} || Key <- Keys],
    UpdateAddressDeps = [{{update_user_address, Key}, [{update_user_address, Key}, {get_user_address, Key}, {page_view, Key}]}
                         || Key <- Keys],
    maps:from_list(PageViewDeps ++ GetAddressDeps ++ UpdateAddressDeps).

-type state() :: #{uid() := zipcode()}.

-spec init() -> state().
init() -> #{}.

-spec update(event(), state(), mailbox()) -> state().
update({{update_user_address, Uid}, {ZipCode, Ts}}, State, SendTo) ->
    SendTo ! {{update_user_address, Uid}, {ZipCode, Ts}},
    maps:put(Uid, ZipCode, State);
update({{page_view, Uid}, Ts}, State, SendTo) ->
    ZipCode = maps:get(Uid, State, 'no_zipcode'),
    State;
update({{get_user_address, Uid}, Ts}, State, SendTo) ->
    ZipCode = maps:get(Uid, State, 'no_zipcode'),
    State.

-spec fork({tag_predicate(), tag_predicate()}, state()) -> {state(), state()}.
fork({Pred1, Pred2}, State) ->
    State1 = maps:filter(fun(K,_) -> is_uid_in_pred(K, Pred1) end, State),
    State2 = maps:filter(fun(K,_) -> is_uid_in_pred(K, Pred2) end, State),
    {State1, State2}.

-spec join(state(), state()) -> state().
join(State1, State2) ->
    util:merge_with(
      fun(K, V1, V2) -> V1 end, State1, State2).

-spec is_uid_in_pred(uid(), tag_predicate()) -> boolean().
is_uid_in_pred(Uid, Pred) ->
    UidTags = tags([Uid]),
    lists:any(Pred, UidTags).

-spec tags(uids()) -> [event_tag()].
tags(Uids) ->
    update_user_address_tags(Uids)
        ++ get_user_address_tags(Uids)
        ++ page_view_tags(Uids).

...
\end{FluminaCode}

\caption{\sys{} implementation of the page-view join example from \Cref{sec:evaluation}.}
\label{fig:page-view-join-erlang}
\end{figure*}

\begin{figure}[H]
  \centering
  \footnotesize{}
\begin{FluminaCode}
-type ab_tag() :: {'a', integer()} | 'b'.
-type event() :: {ab_tag(), integer()}.
-type state() :: {integer(), integer()}.

-define(MODULO, 1000).

depends(ATags) ->
    ADeps = [{ATag, [b]} || ATag <- ATags],
    BDeps = {b, ATags},
    maps:from_list([BDeps|ADeps]).

-spec init() -> state().
init() -> {0, 0}.

-spec update(event(), state(), mailbox()) -> state().
update({{a,Key}, Value}, {Sum, PrevBModulo}, SendTo) ->
    case Value rem ?MODULO == PrevBModulo of
        true ->
            SendTo ! {{a,Key}, Value};
        false ->
            ok
    end,
    {Sum + Value, PrevBModulo};
update({b, Ts}, {Sum, _PrevBModulo}, SendTo) ->
    SendTo ! {sum, {{b, Ts}, Sum}},
    {Sum, (Ts + 1) rem ?MODULO}.

-spec join(state(), state()) -> state().
join({Sum1, PrevBModulo1}, {Sum2, PrevBModulo2}) ->
    {Sum1 + Sum2, PrevBModulo1}.

-spec fork({tag_predicate(), tag_predicate()}, state())
        -> {state(), state()}.
fork(_, {Sum, PrevBModulo}) ->
    {{Sum, PrevBModulo}, {0, PrevBModulo}}.
\end{FluminaCode}
\caption{\sys{} implementation of the Fraud Detection example from \Cref{sec:evaluation}.}
\label{fig:fraud-detection-erlang}
\end{figure}

\section{Timely Dataflow Programs}
\label{appendix:timely-code}

This section contains the real code for the implementations of the three programs in \Cref{sec:evaluation} for Timely Dataflow.

\begin{figure}[H]
  \centering
  \footnotesize{}
\begin{verbatim}
fn event_window_dataflow<G>(
    value_stream: &Stream<G, VBItem>,
    barrier_stream: &Stream<G, VBItem>,
) -> Stream<G, usize>
where
    G: Scope<Timestamp = u128>,
{
    // Use barrier stream to create two clocks, one with
    // heartbeats and one without
    let barrier_broadcast = barrier_stream.broadcast();
    let barrier_clock_withheartbeats = barrier_broadcast
        .map(|_| ());
    let barrier_clock_noheartbeats = barrier_broadcast
        .filter(|x| x.data == VBData::Barrier)
        .map(|_| ());

    // Aggregate values
    value_stream
        .map(|x| match x.data {
            VBData::Value(v) => v,
            VBData::BarrierHeartbeat => unreachable!(),
            VBData::Barrier => unreachable!(),
        })
        .sum()
        .reclock(&barrier_clock_noheartbeats)
        .sum()
        .exchange(|_x| 0)
        .sum()
}
\end{verbatim}
\caption{Timely Dataflow implementation of the Event Windowing example.}
\end{figure}

\begin{figure}[H]
  \centering
  \footnotesize{}

\begin{verbatim}
fn pv_dataflow<G>(
    views: &Stream<G, PVItem>,
    updates: &Stream<G, PVItem>,
) -> Stream<G, PVItem>
where
    G: Scope<Timestamp = u128>,
{
    // partition updates and views by key
    let partitioned_updates = updates.exchange(|x| x.name);
    let partitioned_views = views.exchange(|x| x.name);

    // re-timestamp views using updates
    let updates_clock = partitioned_updates
        .map(|_| ());
    let clocked_views = partitioned_views.reclock(&updates_clock);

    // join each value with the most recent update
    join_by_timestamp(&partitioned_updates, &clocked_views)
        .map(|(x, _y)| x)
}

\end{verbatim}

\caption{Timely Dataflow implementation of the Page-View example, automatic version.}
\end{figure}

\begin{figure}[H]
  \centering
  \footnotesize{}

\begin{verbatim}
fn pv_dataflow_manual<G>(
    views: &Stream<G, PVItem>,
    updates: &Stream<G, PVItem>,
    worker_index: usize,
) -> Stream<G, PVItem>
where
    G: Scope<Timestamp = u128>,
{
    // broadcast updates then filter to only ones relevant to
    // this partition
    // This is low-level because it uses information about
    // physical partitioning of input streams, given by
    // page_partition_function.
    let partitioned_updates = updates
        .broadcast()
        .filter(move |x| {
            x.name == page_partition_function(NUM_PAGES, worker_index)
        });

    // re-timestamp views using updates
    let updates_clock = partitioned_updates.map(|_| ());
    let clocked_views = views.reclock(&updates_clock);

    // join each value with the most recent update
    join_by_timestamp(&partitioned_updates, &clocked_views)
        .map(|(x, _y)| x)
}
\end{verbatim}

\caption{Timely Dataflow implementation of the Page-View example, manual version.}
\end{figure}

\begin{figure}[H]
  \centering
  \footnotesize{}
\begin{verbatim}
    fn fraud_detection_dataflow<G>(
        value_stream: &Stream<G, VBItem>,
        barrier_stream: &Stream<G, VBItem>,
        scope: &mut G,
    ) -> Stream<G, VBItem>
    where
        G: Scope<Timestamp = u128>,
    {
        // Use barrier stream to create two clocks, one with
        // heartbeats and one without
        let barrier_broadcast = barrier_stream.broadcast();
        let barrier_clock_noheartbeats = barrier_broadcast
            .filter(|x| x.data == VBData::Barrier)
            .map(|_| ());

        // Create cycle / feedback loop
        let (handle, model_stream) = scope.feedback(1);
        // start with a model with value 0
        let init_model = (0..1).to_stream(scope);
        let model_broadcast = model_stream
            .concat(&init_model)
            // "line up" the stream with barrier clock, then broadcast
            .reclock(&barrier_clock_noheartbeats)
            .broadcast();

    // Calculate aggregate of values between windows, but
    // now also given the previous state (trained model)
    let value_reclocked = value_stream
        .reclock(&barrier_clock_noheartbeats);

    // Use the model to predict -- includes core logic for
    // how prediction works
    let value_labeled = join_by_timestamp(
        &model_broadcast, &value_reclocked
    ).map(|(model, value)| {
        let label = match value.data {
            VBData::Value(v) => model %
            VBData::BarrierHeartbeat => unreachable!(),
            VBData::Barrier => unreachable!(),
        };
        (value, label)
    });

    // Aggregate the (labeled) values to get the next model
    // In case there are no values in between two barriers,
    // we also need to insert a 0 for each barrier.
    let barrier_zeros = barrier_clock_noheartbeats
        .map(|()| 0);
    value_labeled
        .count()
        .reclock(&barrier_clock_noheartbeats)
        .sum()
        .concat(&barrier_zeros)
        .exchange(|_x| 0)
        .sum()
        .connect_loop(handle);

    // Output labels that were labeled fraudulent
    value_labeled
        .filter(|(_value, label)| *label)
        .map(|(value, _label)| value)
}
\end{verbatim}
\caption{Timely Dataflow implementation of the Fraud Detection example.}
\end{figure}

\onecolumn
\section{Flink Programs}
\label{appendix:flink-code}

In this section we include Flink implementations of Event-Based Windowing,
Page-View Join, and Fraud Detection from \Cref{sec:evaluation}. Since the
implementations are written in Java, which is a fairly verbose programming
language, we omit the code for processing configuration and measuring
performance, data source classes, data classes, and various helper classes.

Since we are omitting the data classes, one of the design patterns we use in the
code may seem unclear, so we explain it before getting to the implementations.
Suppose we have two data classes, \texttt{Value} and \texttt{Heartbeat}, and we want
to have a union type \texttt{ValueOrHeartbeat}. In Java, we can achieve this by
defining \texttt{ValueOrHeartbeat} as an interface implemented by both classes.
Now suppose we have an object \texttt{obj} of type \texttt{ValueOrHeartbeat} and we
would like to implement some logic that depends on the concrete type of \texttt{obj}.
To achieve this, in many other languages we would use pattern matching, but
since Java doesn't have pattern matching built into the language, we implement
it ourselves in the following way.

\begin{lstlisting}[language=Java,basicstyle=\scriptsize\ttfamily]
public interface ValueOrHeartbeat {
    <T> T match(Function<Value, T> valueCase, Function<Heartbeat, T> heartbeatCase);
}

public class Value implements ValueOrHeartbeat {
    @Override
    public <T> T match(Function<Value, T> valueCase, Function<Heartbeat, T> heartbeatCase) {
        return valueCase.apply(this);
    }
}

public class Heartbeat implements ValueOrHeartbeat {
    @Override
    public <T> T match(Function<Value, T> valueCase, Function<Heartbeat, T> heartbeatCase) {
        return heartbeatCase.apply(this);
    }
}
\end{lstlisting}

\noindent The interface defines a method \texttt{match} parametric in a type
parameter \texttt{T} which takes two ``case'' functions as arguments: one taking
\texttt{Value}, the other taking \texttt{Heartbeat} as input, and both functions
returning \texttt{T} as output. The concrete classes implement \texttt{match} by
calling the appropriate ``case'' function with their object instance \texttt{this}
as input and returning the ``case'' function's output.

As an example of common usage of pattern matching in the code below, suppose we
have a list called \texttt{values} of type \texttt{List<Value>} and an object called
\texttt{valueOrHeartbeat} of type \texttt{ValueOrHeartbeat}. We want to add the
object to the list if it is a value. We can do this with a single line of code:

\begin{lstlisting}[language=Java,basicstyle=\scriptsize\ttfamily]
values.addAll(valueOrHeartbeat.match(v -> List.of(v), hb -> Collections.emptyList()));
\end{lstlisting}

\noindent By writing the value case more succinctly in the point-free style, we
get to the kind of line that commonly occurs in our code:

\begin{lstlisting}[language=Java,basicstyle=\scriptsize\ttfamily]
values.addAll(valueOrHeartbeat.match(List::of, hb -> Collections.emptyList()));
\end{lstlisting}

\subsection{Event-based windowing}

In this example, the inputs are a value stream consisting of several parallel
substreams and a single barrier stream. The number of parallel value substreams
is for historical reasons called \texttt{conf.getValueNodes()}. Both values and
barriers have timestamps, obtained by calling the method \texttt{getPhysicalTimestamp()},
which define how the values and barriers should be interleaved. The
task is to calculate the sum of values between two consecutive barriers.

\subsubsection{Sequential implementation}
In the sequential version, we perform a low-level join of the two streams using
Flink's \texttt{connect} operator. The logic of the join is implemented in the
class \texttt{ValueBarrierProcessSequential} listed below. Since the class
maintains non-trivial internal state, it cannot be parallelized by simply
instantiating multiple copies.

\begin{lstlisting}[language=Java,basicstyle=\tiny\ttfamily,commentstyle=\tiny\ttfamily,morekeywords={var}]
public class ValueBarrierSequentialExperiment implements Experiment {

    private final ValueBarrierConfig conf;

    public ValueBarrierSequentialExperiment(final ValueBarrierConfig conf) {
        this.conf = conf;
    }

    @Override
    public JobExecutionResult run(final StreamExecutionEnvironment env, final Instant startTime) throws Exception {
        env.setParallelism(1);

        final var valueSource = new ValueOrHeartbeatSource(conf.getTotalValues(), conf.getValueRate(), startTime);
        final var valueStream = env.addSource(valueSource)
                .setParallelism(conf.getValueNodes())
                .slotSharingGroup("values");
        final var barrierSource = new BarrierOrHeartbeatSource(
                conf.getTotalValues(), conf.getValueRate(), conf.getValueBarrierRatio(),
                conf.getHeartbeatRatio(), startTime);
        final var barrierStream = env.addSource(barrierSource)
                .slotSharingGroup("barriers");

        valueStream.connect(barrierStream)
                .process(new ValueBarrierProcessSequential(conf.getValueNodes()))
                .slotSharingGroup("barriers")
                .map(new TimestampMapper())
                .writeAsText(conf.getOutFile(), FileSystem.WriteMode.OVERWRITE);

        return env.execute("ValueBarrier Experiment");
    }

    @Override
    public long getTotalEvents() {
        return conf.getValueNodes() * conf.getTotalValues() + conf.getTotalValues() / conf.getValueBarrierRatio();
    }

    @Override
    public long getOptimalThroughput() {
        return (long) (conf.getValueRate() * conf.getValueNodes() + conf.getValueRate() / conf.getValueBarrierRatio());
    }

}

public class ValueBarrierProcessSequential extends
        CoProcessFunction<ValueOrHeartbeat, BarrierOrHeartbeat, Tuple3<Long, Long, Instant>> {

    private final List<Instant> valueTimestamps = new ArrayList<>();
    private final PriorityQueue<Value> values = new PriorityQueue<>(new TimestampComparator());
    private final Queue<Barrier> barriers = new ArrayDeque<>();
    private Instant barrierTimestamp = Instant.MIN;
    private long sum = 0L;

    public ValueBarrierProcessSequential(final int valueParallelism) {
        valueTimestamps.addAll(Collections.nCopies(valueParallelism, Instant.MIN));
    }

    @Override
    public void processElement1(final ValueOrHeartbeat valueOrHeartbeat,
                                final Context ctx,
                                final Collector<Tuple3<Long, Long, Instant>> out) {
        values.addAll(valueOrHeartbeat.match(List::of, hb -> Collections.emptyList()));
        valueTimestamps.set(valueOrHeartbeat.getSourceIndex(), valueOrHeartbeat.getPhysicalTimestamp());
        makeProgress(out);
    }

    @Override
    public void processElement2(final BarrierOrHeartbeat barrierOrHeartbeat,
                                final Context ctx,
                                final Collector<Tuple3<Long, Long, Instant>> out) {
        barriers.addAll(barrierOrHeartbeat.match(List::of, hb -> Collections.emptyList()));
        barrierTimestamp = barrierOrHeartbeat.getPhysicalTimestamp();
        makeProgress(out);
    }

    private Instant getCurrentTimestamp() {
        final var valueTimestamp = valueTimestamps.stream().min(Instant::compareTo).get();
        return min(valueTimestamp, barrierTimestamp);
    }

    private void makeProgress(final Collector<Tuple3<Long, Long, Instant>> out) {
        final var currentTimestamp = getCurrentTimestamp();
        while (!barriers.isEmpty() &&
                barriers.element().getPhysicalTimestamp().compareTo(currentTimestamp) <= 0) {
            final var barrier = barriers.remove();
            while (!values.isEmpty() &&
                    values.element().getPhysicalTimestamp().isBefore(barrier.getPhysicalTimestamp())) {
                update(values.remove());
            }
            update(barrier, out);
        }
        while (!values.isEmpty() &&
                values.element().getPhysicalTimestamp().compareTo(currentTimestamp) <= 0) {
            update(values.remove());
        }
    }

    private void update(final Value value) {
        sum += value.val;
    }

    private void update(final Barrier barrier,
                        final Collector<Tuple3<Long, Long, Instant>> out) {
        out.collect(Tuple3.of(sum, barrier.getLogicalTimestamp(), barrier.getPhysicalTimestamp()));
        sum = 0L;
    }

}
\end{lstlisting}

\subsubsection{Parallel implementation}
The parallel implementation uses Flink's broadcast state pattern. The barriers
are broadcast to all parallel value substreams, which can then compute their
partial sums. The partial sums are later aggregated into the total sum. We use a
trick of converting barriers to Flink's watermarks, which allows us to define
logical windows over which we perform aggregation.

\begin{lstlisting}[language=Java,basicstyle=\tiny\ttfamily,commentstyle=\tiny\ttfamily,morekeywords={var}]
public class ValueBarrierExperiment implements Experiment {

    private final ValueBarrierConfig conf;

    public ValueBarrierExperiment(final ValueBarrierConfig conf) {
        this.conf = conf;
    }

    @Override
    public JobExecutionResult run(final StreamExecutionEnvironment env, final Instant startTime) throws Exception {
        env.setParallelism(1);
        env.setStreamTimeCharacteristic(TimeCharacteristic.EventTime);

        final var valueSource = new ValueOrHeartbeatSource(conf.getTotalValues(), conf.getValueRate(), startTime);
        final var valueStream = env.addSource(valueSource)
                .setParallelism(conf.getValueNodes())
                .slotSharingGroup("values");
        final var barrierSource = new BarrierOrHeartbeatSource(
                conf.getTotalValues(), conf.getValueRate(), conf.getValueBarrierRatio(),
                conf.getHeartbeatRatio(), startTime);
        final var barrierStream = env.addSource(barrierSource)
                .slotSharingGroup("barriers");

        // Broadcast the barrier stream and connect it with the value stream
        // We use a dummy broadcast state descriptor that is never actually used.
        final var broadcastStateDescriptor =
                new MapStateDescriptor<>("BroadcastState", Void.class, Void.class);
        final var broadcastStream = barrierStream.broadcast(broadcastStateDescriptor);

        valueStream.connect(broadcastStream)
                .process(new BroadcastProcessFunction<ValueOrHeartbeat, BarrierOrHeartbeat,
                        Tuple3<Long, Long, Instant>>() {
                    private Instant valuePhysicalTimestamp = Instant.MIN;
                    private Instant barrierPhysicalTimestamp = Instant.MIN;
                    private long sum = 0;

                    private final Queue<Value> unprocessedValues = new ArrayDeque<>();
                    private final Queue<Barrier> unprocessedBarriers = new ArrayDeque<>();

                    @Override
                    public void processElement(final ValueOrHeartbeat valueOrHeartbeat,
                                               final ReadOnlyContext ctx,
                                               final Collector<Tuple3<Long, Long, Instant>> collector) {
                        unprocessedValues.addAll(valueOrHeartbeat.match(List::of, hb -> Collections.emptyList()));
                        valuePhysicalTimestamp = valueOrHeartbeat.getPhysicalTimestamp();
                        makeProgress(collector);
                    }

                    @Override
                    public void processBroadcastElement(final BarrierOrHeartbeat barrierOrHeartbeat,
                                                        final Context ctx,
                                                        final Collector<Tuple3<Long, Long, Instant>> collector) {
                        unprocessedBarriers.addAll(barrierOrHeartbeat.match(List::of, hb -> Collections.emptyList()));
                        barrierPhysicalTimestamp = barrierOrHeartbeat.getPhysicalTimestamp();
                        makeProgress(collector);
                    }

                    private void makeProgress(final Collector<Tuple3<Long, Long, Instant>> collector) {
                        final var currentTime = min(valuePhysicalTimestamp, barrierPhysicalTimestamp);
                        while (!unprocessedValues.isEmpty() &&
                                unprocessedValues.element().getPhysicalTimestamp().compareTo(currentTime) <= 0) {
                            final var value = unprocessedValues.remove();
                            while (!unprocessedBarriers.isEmpty() &&
                                    unprocessedBarriers.element()
                                        .getPhysicalTimestamp().isBefore(value.getPhysicalTimestamp())) {
                                update(unprocessedBarriers.remove(), collector);
                            }
                            update(value, collector);
                        }
                        while (!unprocessedBarriers.isEmpty() &&
                                unprocessedBarriers.element().getPhysicalTimestamp().compareTo(currentTime) <= 0) {
                            update(unprocessedBarriers.remove(), collector);
                        }
                    }

                    private void update(final Value value, final Collector<Tuple3<Long, Long, Instant>> collector) {
                        sum += value.val;
                    }

                    private void update(final Barrier barrier, final Collector<Tuple3<Long, Long, Instant>> collector) {
                        collector.collect(Tuple3.of(sum, barrier.getLogicalTimestamp(),
                            barrier.getPhysicalTimestamp()));
                        sum = 0;
                    }
                })
                .setParallelism(conf.getValueNodes())
                .slotSharingGroup("values")
                .assignTimestampsAndWatermarks(new AssignerWithPunctuatedWatermarks<Tuple3<Long, Long, Instant>>() {
                    @Override
                    public Watermark checkAndGetNextWatermark(final Tuple3<Long, Long, Instant> tuple, final long l) {
                        return new Watermark(l);
                    }

                    @Override
                    public long extractTimestamp(final Tuple3<Long, Long, Instant> tuple, final long l) {
                        // The field tuple.f1 corresponds to the logical timestamp
                        return tuple.f1;
                    }
                })
                .setParallelism(conf.getValueNodes())
                .timeWindowAll(Time.milliseconds(conf.getValueBarrierRatio()))
                .reduce((x, y) -> {
                    x.f0 += y.f0;
                    return x;
                })
                .slotSharingGroup("barriers")
                .map(new TimestampMapper())
                .writeAsText(conf.getOutFile(), FileSystem.WriteMode.OVERWRITE);

        return env.execute("ValueBarrier Experiment");
    }

    @Override
    public long getTotalEvents() {
        return conf.getValueNodes() * conf.getTotalValues() + conf.getTotalValues() / conf.getValueBarrierRatio();
    }

    @Override
    public long getOptimalThroughput() {
        return (long) (conf.getValueRate() * conf.getValueNodes() + conf.getValueRate() / conf.getValueBarrierRatio());
    }

}
\end{lstlisting}

\subsubsection{Manual implementation} In the manual implementation, we simulate
\sys{}'s fork-join requests and responses by having the value and barrier
processing operators (classes \texttt{ValueProcessManual} and
\texttt{BarrierProcessManual}) communicate over an external service (class
\texttt{ValueBarrierService}) using Java RMI (Remote Method Invocation). The service
implements synchronization using semaphores as described in \Cref{sec:evaluation}.

\begin{lstlisting}[language=Java,basicstyle=\tiny\ttfamily,commentstyle=\tiny\ttfamily,morekeywords={var}]
public class ValueBarrierManualExperiment implements Experiment {

    private final ValueBarrierConfig conf;

    public ValueBarrierManualExperiment(final ValueBarrierConfig conf) {
        this.conf = conf;
    }

    @Override
    public JobExecutionResult run(final StreamExecutionEnvironment env, final Instant startTime) throws Exception {
        env.setParallelism(1);

        // Set up the remote ForkJoin service
        final var valueBarrierService = new ValueBarrierService(conf.getValueNodes());
        @SuppressWarnings("unchecked") final var valueBarrierServiceStub =
                (ForkJoinService<Long, Long>) UnicastRemoteObject.exportObject(valueBarrierService, 0);
        final var valueBarrierServiceName = UUID.randomUUID().toString();
        final var registry = LocateRegistry.getRegistry(conf.getRmiHost());
        registry.rebind(valueBarrierServiceName, valueBarrierServiceStub);

        final var valueSource = new ValueOrHeartbeatSource(conf.getTotalValues(), conf.getValueRate(), startTime);
        final var valueStream = env.addSource(valueSource)
                .setParallelism(conf.getValueNodes())
                .slotSharingGroup("values");
        final var barrierSource = new BarrierOrHeartbeatSource(
                conf.getTotalValues(), conf.getValueRate(), conf.getValueBarrierRatio(),
                conf.getHeartbeatRatio(), startTime);
        final var barrierStream = env.addSource(barrierSource)
                .slotSharingGroup("barriers");

        // Broadcast the barrier stream and connect it with the value stream
        // We use a dummy broadcast state descriptor that is never actually used.
        final var broadcastStateDescriptor =
                new MapStateDescriptor<>("BroadcastState", Void.class, Void.class);
        final var broadcastStream = barrierStream.broadcast(broadcastStateDescriptor);

        valueStream.connect(broadcastStream)
                .process(new ValueProcessManual(conf.getRmiHost(), valueBarrierServiceName))
                .setParallelism(conf.getValueNodes())
                .slotSharingGroup("values");

        barrierStream
                .flatMap(new FlatMapFunction<BarrierOrHeartbeat, Barrier>() {
                    @Override
                    public void flatMap(final BarrierOrHeartbeat barrierOrHeartbeat, final Collector<Barrier> out) {
                        barrierOrHeartbeat.match(
                                barrier -> {
                                    out.collect(barrier);
                                    return null;
                                },
                                heartbeat -> null
                        );
                    }
                })
                .process(new BarrierProcessManual(conf.getRmiHost(), valueBarrierServiceName))
                .map(new TimestampMapper())
                .writeAsText(conf.getOutFile(), FileSystem.WriteMode.OVERWRITE);

        try {
            return env.execute("ValueBarrier Manual Experiment");
        } finally {
            UnicastRemoteObject.unexportObject(valueBarrierService, true);
            try {
                registry.unbind(valueBarrierServiceName);
            } catch (final NotBoundException ignored) {

            }
        }
    }

    @Override
    public long getTotalEvents() {
        return conf.getValueNodes() * conf.getTotalValues() + conf.getTotalValues() / conf.getValueBarrierRatio();
    }

    @Override
    public long getOptimalThroughput() {
        return (long) (conf.getValueRate() * conf.getValueNodes() + conf.getValueRate() / conf.getValueBarrierRatio());
    }

}

public class ValueProcessManual extends BroadcastProcessFunction<ValueOrHeartbeat, BarrierOrHeartbeat, Void> {

    private Instant valuePhysicalTimestamp = Instant.MIN;
    private Instant barrierPhysicalTimestamp = Instant.MIN;
    private long sum = 0;

    private final Queue<Value> unprocessedValues = new ArrayDeque<>();
    private final Queue<Barrier> unprocessedBarriers = new ArrayDeque<>();

    private final String rmiHost;
    private final String valueBarrierServiceName;
    private transient ForkJoinService<Long, Long> valueBarrierService;

    public ValueProcessManual(final String rmiHost, final String valueBarrierServiceName) {
        this.rmiHost = rmiHost;
        this.valueBarrierServiceName = valueBarrierServiceName;
    }

    @Override
    @SuppressWarnings("unchecked")
    public void open(final Configuration parameters) throws RemoteException, NotBoundException {
        final var registry = LocateRegistry.getRegistry(rmiHost);
        valueBarrierService = (ForkJoinService<Long, Long>) registry.lookup(valueBarrierServiceName);
    }

    @Override
    public void processElement(final ValueOrHeartbeat valueOrHeartbeat,
                               final ReadOnlyContext ctx,
                               final Collector<Void> out) throws RemoteException {
        unprocessedValues.addAll(valueOrHeartbeat.match(List::of, hb -> Collections.emptyList()));
        valuePhysicalTimestamp = valueOrHeartbeat.getPhysicalTimestamp();
        makeProgress();
    }

    @Override
    public void processBroadcastElement(final BarrierOrHeartbeat barrierOrHeartbeat,
                                        final Context ctx,
                                        final Collector<Void> out) throws RemoteException {
        unprocessedBarriers.addAll(barrierOrHeartbeat.match(List::of, hb -> Collections.emptyList()));
        barrierPhysicalTimestamp = barrierOrHeartbeat.getPhysicalTimestamp();
        makeProgress();
    }

    private void makeProgress() throws RemoteException {
        final var currentTime = min(valuePhysicalTimestamp, barrierPhysicalTimestamp);
        while (!unprocessedValues.isEmpty() &&
                unprocessedValues.element().getPhysicalTimestamp().compareTo(currentTime) <= 0) {
            final var value = unprocessedValues.remove();
            while (!unprocessedBarriers.isEmpty() &&
                    unprocessedBarriers.element().getPhysicalTimestamp().isBefore(value.getPhysicalTimestamp())) {
                join(unprocessedBarriers.remove());
            }
            update(value);
        }
        while (!unprocessedBarriers.isEmpty() &&
                unprocessedBarriers.element().getPhysicalTimestamp().compareTo(currentTime) <= 0) {
            join(unprocessedBarriers.remove());
        }
    }

    private void update(final Value value) {
        sum += value.val;
    }

    private void join(final Barrier barrier) throws RemoteException {
        sum = valueBarrierService.joinChild(getRuntimeContext().getIndexOfThisSubtask(), sum);
    }

}

public class BarrierProcessManual extends ProcessFunction<Barrier, Tuple3<Long, Long, Instant>> {

    private final String rmiHost;
    private final String valueBarrierServiceName;
    private transient ForkJoinService<Long, Long> valueBarrierService;
    private long sum;

    public BarrierProcessManual(final String rmiHost, final String valueBarrierServiceName) {
        this.rmiHost = rmiHost;
        this.valueBarrierServiceName = valueBarrierServiceName;
        this.sum = 0L;
    }

    @Override
    @SuppressWarnings("unchecked")
    public void open(final Configuration parameters) throws RemoteException, NotBoundException {
        final var registry = LocateRegistry.getRegistry(rmiHost);
        valueBarrierService = (ForkJoinService<Long, Long>) registry.lookup(valueBarrierServiceName);
    }

    @Override
    public void processElement(final Barrier barrier,
                               final Context ctx,
                               final Collector<Tuple3<Long, Long, Instant>> out) throws RemoteException {
        sum = valueBarrierService.joinParent(getRuntimeContext().getIndexOfThisSubtask(), sum);
        out.collect(Tuple3.of(sum, barrier.getLogicalTimestamp(), barrier.getPhysicalTimestamp()));
    }

}

public class ValueBarrierService implements ForkJoinService<Long, Long> {

    private final int valueParallelism;
    private final List<Semaphore> joinSemaphores;
    private final List<Semaphore> forkSemaphores;
    private final long[] states;

    public ValueBarrierService(final int valueParallelism) {
        this.valueParallelism = valueParallelism;
        joinSemaphores = IntStream.range(0, valueParallelism)
                .mapToObj(x -> new Semaphore(0)).collect(Collectors.toList());
        forkSemaphores = IntStream.range(0, valueParallelism)
                .mapToObj(x -> new Semaphore(0)).collect(Collectors.toList());
        states = new long[valueParallelism];
    }

    @Override
    public Long joinChild(final int subtaskIndex, final Long state) {
        states[subtaskIndex] = state;
        joinSemaphores.get(subtaskIndex).release();
        forkSemaphores.get(subtaskIndex).acquireUninterruptibly();
        return states[subtaskIndex];
    }

    @Override
    public Long joinParent(final int subtaskIndex, final Long state) {
        long sum = 0L;
        for (int i = 0; i < valueParallelism; ++i) {
            joinSemaphores.get(i).acquireUninterruptibly();
            sum += states[i];
        }
        // Simulate fork propagation in the opposite direction
        for (int i = valueParallelism - 1; i >= 0; --i) {
            states[i] = 0;
            forkSemaphores.get(i).release();
        }
        return sum;
    }

}
\end{lstlisting}

\subsection{Page-view join}

In this example, the inputs are a stream of get and update events, of which only
the update events are relevant, and a stream of page-view events consisting of
a number of parallel substreams; the number of substreams is
given by \texttt{conf.getPageViewParallelism()} in the code below. An update
event updates some metadata (a user's ZIP code in the code below), and the processing
of page-view events depends on the most recent metadata. Hence, the task is to
correctly join the two streams while achieving the best possible parallelism.

\subsubsection{Sequential implementation}
As in the previous example, the sequential implementation here is again a
straightforward low-level join using the Flink's \texttt{connect} operator. The
logic of the join is implemented in the
class \texttt{PageViewProcessSequential}, which again cannot be parallelized by
simply instantiating multiple copies.

\begin{lstlisting}[language=Java,basicstyle=\tiny\ttfamily,commentstyle=\tiny\ttfamily,morekeywords={var}]
public class PageViewSequentialExperiment implements Experiment {

    private final PageViewConfig conf;

    public PageViewSequentialExperiment(final PageViewConfig conf) {
        this.conf = conf;
    }

    @Override
    public JobExecutionResult run(final StreamExecutionEnvironment env, final Instant startTime) throws Exception {
        env.setParallelism(1);

        final var getOrUpdateSource = new GetOrUpdateOrHeartbeatSource(conf.getTotalPageViews(),
                conf.getTotalUsers(), conf.getPageViewRate(), startTime);
        final var getOrUpdateStream = env.addSource(getOrUpdateSource)
                .slotSharingGroup("getOrUpdate");
        final var pageViewSource = new PageViewOrHeartbeatSource(conf.getTotalPageViews(),
                conf.getTotalUsers(), conf.getPageViewRate(), startTime);
        final var pageViewStream = env.addSource(pageViewSource)
                .setParallelism(conf.getPageViewParallelism());

        getOrUpdateStream.connect(pageViewStream)
                .process(new PageViewProcessSequential(conf.getTotalUsers(), conf.getPageViewParallelism()))
                .map(new TimestampMapper())
                .writeAsText(conf.getOutFile(), FileSystem.WriteMode.OVERWRITE);

        return env.execute("PageView Experiment");
    }

    @Override
    public long getTotalEvents() {
        // PageView events + Get events + Update events
        return (conf.getTotalPageViews() * conf.getPageViewParallelism() +
                conf.getTotalPageViews() / 100 + conf.getTotalPageViews() / 1000) * conf.getTotalUsers();
    }

    @Override
    public long getOptimalThroughput() {
        return (long) ((conf.getPageViewParallelism() + 0.011) * conf.getTotalUsers() * conf.getPageViewRate());
    }

}

public class PageViewProcessSequential extends CoProcessFunction<GetOrUpdateOrHeartbeat, PageViewOrHeartbeat, Update> {

    private final List<Instant> pageViewTimestamps = new ArrayList<>();
    private final List<Integer> zipCodes = new ArrayList<>();
    private final Queue<Update> updates = new ArrayDeque<>();
    private final PriorityQueue<PageView> pageViews = new PriorityQueue<>(new TimestampComparator());
    private Instant updateTimestamp = Instant.MIN;

    public PageViewProcessSequential(final int totalUsers, final int pageViewParallelism) {
        pageViewTimestamps.addAll(Collections.nCopies(pageViewParallelism, Instant.MIN));
        zipCodes.addAll(Collections.nCopies(totalUsers, 10_000));
    }

    @Override
    public void processElement1(final GetOrUpdateOrHeartbeat getOrUpdateOrHeartbeat,
                                final Context ctx,
                                final Collector<Update> out) {
        updates.addAll(getOrUpdateOrHeartbeat.match(
                gou -> gou.match(g -> Collections.emptyList(), List::of),
                hb -> Collections.emptyList()));
        updateTimestamp = getOrUpdateOrHeartbeat.getPhysicalTimestamp();
        makeProgress(out);
    }

    @Override
    public void processElement2(final PageViewOrHeartbeat pageViewOrHeartbeat,
                                final Context ctx,
                                final Collector<Update> out) {
        pageViews.addAll(pageViewOrHeartbeat.match(List::of, hb -> Collections.emptyList()));
        pageViewTimestamps.set(pageViewOrHeartbeat.getSourceIndex(), pageViewOrHeartbeat.getPhysicalTimestamp());
        makeProgress(out);
    }

    private Instant getCurrentTimestamp() {
        final var pageViewTimestamp = pageViewTimestamps.stream().min(Instant::compareTo).get();
        return min(pageViewTimestamp, updateTimestamp);
    }

    private void makeProgress(final Collector<Update> out) {
        final var currentTimestamp = getCurrentTimestamp();
        while (!updates.isEmpty() &&
                updates.element().getPhysicalTimestamp().compareTo(currentTimestamp) <= 0) {
            final var update = updates.remove();
            while (!pageViews.isEmpty() &&
                    pageViews.element().getPhysicalTimestamp().isBefore(update.getPhysicalTimestamp())) {
                update(pageViews.remove(), out);
            }
            update(update, out);
        }
        while (!pageViews.isEmpty() &&
                pageViews.element().getPhysicalTimestamp().compareTo(currentTimestamp) <= 0) {
            update(pageViews.remove(), out);
        }
    }

    private void update(final Update update, final Collector<Update> out) {
        zipCodes.set(update.getUserId(), update.zipCode);
        out.collect(update);
    }

    private void update(final PageView pageView, final Collector<Update> out) {
        // This update is a no-op
    }

}
\end{lstlisting}

\subsubsection{Parallel implementation}
\label{sssec:page-view-parallel}

In the parallel implementation, we first partition the streams by user ID, and
then perform the join. Thus, the achieved parallelism corresponds to the total
number of users, provided that the events are uniformly distributed over all
users, which is the case in our experiment.
The logic of the join is implemented in the class \texttt{PageViewProcessParallel}.
Here we have multiple subtasks, each independently processing events related to
a single user.

In order for the events to be routed to a correct \texttt{PageViewProcessParallel} subtask,
we had to resort to a couple of low-level tricks. First, in order to partition
the update stream, we also need to partition the heartbeats occurring in that
stream. But since the heartbeats are originally not tied to a particular user,
we need to duplicate them for each user ID. Second, the hash function used by
Flink for key-based partitioning is not injective when restricted to the domain
of our user IDs. Thus, page-view events related to multiple users may end up on
the same \texttt{PageViewProcessParallel} subtask. To prevent this, we would like
to override the hash function, but since Flink does not allow this directly, we
do it by inverting the function and using as keys the inverse of user IDs. The
achieved effect is that the function is composed with its inverse and thus
canceled out.

\begin{lstlisting}[language=Java,basicstyle=\tiny\ttfamily,commentstyle=\tiny\ttfamily,morekeywords={var}]
public class PageViewExperiment implements Experiment {

    private final PageViewConfig conf;

    public PageViewExperiment(final PageViewConfig conf) {
        this.conf = conf;
    }

    @Override
    public JobExecutionResult run(final StreamExecutionEnvironment env, final Instant startTime) throws Exception {
        env.setParallelism(1);

        final var getOrUpdateSource = new GetOrUpdateOrHeartbeatSource(conf.getTotalPageViews(),
                conf.getTotalUsers(), conf.getPageViewRate(), startTime);
        final var getOrUpdateStream = env.addSource(getOrUpdateSource)
                .slotSharingGroup("getOrUpdate");
        final var pageViewSource = new PageViewOrHeartbeatSource(conf.getTotalPageViews(),
                conf.getTotalUsers(), conf.getPageViewRate(), startTime);
        final var pageViewStream = env.addSource(pageViewSource)
                .setParallelism(conf.getPageViewParallelism());

        // We need to duplicate the GetOrUpdate heartbeats, so that each processing instance gets a copy
        final var totalUsers = conf.getTotalUsers();
        final var gouWithDupHeartbeats = getOrUpdateStream
                .flatMap(new FlatMapFunction<GetOrUpdateOrHeartbeat, GetOrUpdateOrHeartbeat>() {

                    @Override
                    public void flatMap(final GetOrUpdateOrHeartbeat getOrUpdateOrHeartbeat,
                                        final Collector<GetOrUpdateOrHeartbeat> out) {
                        final List<GetOrUpdateOrHeartbeat> toCollect = getOrUpdateOrHeartbeat.match(
                                gou -> gou.match(List::of, List::of),
                                hb -> IntStream.range(0, totalUsers)
                                        .mapToObj(userId -> new GetOrUpdateHeartbeat(hb, userId))
                                        .collect(Collectors.toList())
                        );
                        toCollect.forEach(out::collect);
                    }
                });

        // Normal low-level join
        // We invert the key so that each event is routed to a correct parallel processing instance
        final var invertedUserIds = FlinkHashInverter.getMapping(conf.getTotalUsers());
        gouWithDupHeartbeats.keyBy(gou -> invertedUserIds.get(gou.getUserId()))
                .connect(pageViewStream.keyBy(pv -> invertedUserIds.get(pv.getUserId())))
                .process(new PageViewProcessParallel(conf.getPageViewParallelism()))
                .setParallelism(conf.getTotalUsers())
                .map(new TimestampMapper())
                .setParallelism(conf.getTotalUsers())
                .writeAsText(conf.getOutFile(), FileSystem.WriteMode.OVERWRITE);

        return env.execute("PageView Experiment");
    }

    @Override
    public long getTotalEvents() {
        // PageView events + Get events + Update events
        return (conf.getTotalPageViews() * conf.getPageViewParallelism() +
                conf.getTotalPageViews() / 100 + conf.getTotalPageViews() / 1000) * conf.getTotalUsers();
    }

    @Override
    public long getOptimalThroughput() {
        return (long) ((conf.getPageViewParallelism() + 0.011) * conf.getTotalUsers() * conf.getPageViewRate());
    }

}

public class PageViewProcessParallel extends CoProcessFunction<GetOrUpdateOrHeartbeat, PageViewOrHeartbeat, Update> {

    private final List<Instant> pageViewTimestamps = new ArrayList<>();
    private final Queue<Update> updates = new ArrayDeque<>();
    private final PriorityQueue<PageView> pageViews = new PriorityQueue<>(new TimestampComparator());
    private Instant updateTimestamp = Instant.MIN;

    private int zipCode = 10_000;

    public PageViewProcessParallel(final int pageViewParallelism) {
        pageViewTimestamps.addAll(Collections.nCopies(pageViewParallelism, Instant.MIN));
    }

    @Override
    public void processElement1(final GetOrUpdateOrHeartbeat getOrUpdateOrHeartbeat,
                                final Context ctx,
                                final Collector<Update> out) {
        updates.addAll(getOrUpdateOrHeartbeat.match(
                gou -> gou.match(g -> Collections.emptyList(), List::of),
                hb -> Collections.emptyList()));
        updateTimestamp = getOrUpdateOrHeartbeat.getPhysicalTimestamp();
        makeProgress(out);
    }

    @Override
    public void processElement2(final PageViewOrHeartbeat pageViewOrHeartbeat,
                                final Context ctx,
                                final Collector<Update> out) {
        pageViews.addAll(pageViewOrHeartbeat.match(List::of, hb -> Collections.emptyList()));
        pageViewTimestamps.set(pageViewOrHeartbeat.getSourceIndex(), pageViewOrHeartbeat.getPhysicalTimestamp());
        makeProgress(out);
    }

    private Instant getCurrentTimestamp() {
        final var pageViewTimestamp = pageViewTimestamps.stream().min(Instant::compareTo).get();
        return min(pageViewTimestamp, updateTimestamp);
    }

    private void makeProgress(final Collector<Update> out) {
        final var currentTimestamp = getCurrentTimestamp();
        while (!updates.isEmpty() &&
                updates.element().getPhysicalTimestamp().compareTo(currentTimestamp) <= 0) {
            final var update = updates.remove();
            while (!pageViews.isEmpty() &&
                    pageViews.element().getPhysicalTimestamp().isBefore(update.getPhysicalTimestamp())) {
                update(pageViews.remove(), out);
            }
            update(update, out);
        }
        while (!pageViews.isEmpty() &&
                pageViews.element().getPhysicalTimestamp().compareTo(currentTimestamp) <= 0) {
            update(pageViews.remove(), out);
        }
    }

    private void update(final Update update, final Collector<Update> out) {
        zipCode = update.zipCode;
        out.collect(update);
    }

    private void update(final PageView pageView, final Collector<Update> out) {
        // This update is a no-op
    }

}
\end{lstlisting}

\subsubsection{Manual implementation}

In the manual implementation we broadcast the updates over $p$ parallel page-view
substreams ($p$ is \texttt{pageViewParallelism} in the code below),
which are additionally partitioned by (inverted) user IDs. The resulting
stream is processed by \texttt{PageViewProcessManual}, which can now have $p\cdot n$
parallel subtasks for $n$ users. The update events are also partitioned by
(inverted) user IDs and processed by $n$ parallel subtasks of
\texttt{UpdateProcessManual}. To process updates, we ``pretend'' that an
\texttt{UpdateProcessManual} subtask needs to join with $p$
\texttt{PageViewProcessManual} subtasks that process page views related to the same
user ID. The fork-join requests and responses are implemented as in the previous
example using an external service (called \texttt{PageViewService} in the code below).

\begin{lstlisting}[language=Java,basicstyle=\tiny\ttfamily,commentstyle=\tiny\ttfamily,morekeywords={var}]
public class PageViewManualExperiment implements Experiment {

    private final PageViewConfig conf;

    public PageViewManualExperiment(final PageViewConfig conf) {
        this.conf = conf;
    }

    @Override
    public JobExecutionResult run(final StreamExecutionEnvironment env, final Instant startTime) throws Exception {
        env.setParallelism(1);

        final int pageViewParallelism = conf.getPageViewParallelism();

        // Set up the remote ForkJoin service
        final var pageViewService = new PageViewService(conf.getTotalUsers(), pageViewParallelism);
        @SuppressWarnings("unchecked") final var pageViewServiceStub =
                (ForkJoinService<Tuple2<Long, Long>, Tuple2<Long, Long>>)
                UnicastRemoteObject.exportObject(pageViewService, 0);
        final var pageViewServiceName = UUID.randomUUID().toString();
        final var registry = LocateRegistry.getRegistry(conf.getRmiHost());
        registry.rebind(pageViewServiceName, pageViewServiceStub);

        final var getOrUpdateSource = new GetOrUpdateOrHeartbeatSource(conf.getTotalPageViews(),
                conf.getTotalUsers(), conf.getPageViewRate(), startTime);
        final var getOrUpdateStream = env.addSource(getOrUpdateSource)
                .slotSharingGroup("getOrUpdate");
        final var pageViewSource = new PageViewOrHeartbeatSource(conf.getTotalPageViews(),
                conf.getTotalUsers(), conf.getPageViewRate(), startTime);
        final var pageViewStream = env.addSource(pageViewSource)
                .setParallelism(pageViewParallelism);

        final var broadcastStateDescriptor = new MapStateDescriptor<>("Dummy", Void.class, Void.class);
        final var broadcastGetOrUpdateStream =
                getOrUpdateStream.broadcast(broadcastStateDescriptor);

        // We use the key inverter to correctly distribute events over parallel instances. There are
        // totalUsers * pageViewParallelism processing subtasks, and we want a pageView with given userId
        // and subtaskIndex (determined by the pageViewStream instance) to be processed by subtask with
        // with index userId * pageViewParallelism + subtaskIndex. The inverter inverts this index to a
        // key that, once hashed, will map back to the index.
        final var keyInverter =
                FlinkHashInverter.getMapping(conf.getTotalUsers() * pageViewParallelism);
        pageViewStream.keyBy(poh -> keyInverter.get(poh.getUserId() * pageViewParallelism + poh.getSourceIndex()))
                .connect(broadcastGetOrUpdateStream)
                .process(new PageViewProcessManual(conf.getRmiHost(), pageViewServiceName, pageViewParallelism))
                .setParallelism(conf.getTotalUsers() * pageViewParallelism);

        // We again use an inverter to invert user IDs
        final var invertedUserIds = FlinkHashInverter.getMapping(conf.getTotalUsers());
        getOrUpdateStream
                .flatMap(new FlatMapFunction<GetOrUpdateOrHeartbeat, Update>() {
                    @Override
                    public void flatMap(final GetOrUpdateOrHeartbeat getOrUpdateOrHeartbeat,
                                        final Collector<Update> out) {
                        getOrUpdateOrHeartbeat.match(
                                getOrUpdate -> getOrUpdate.match(
                                        get -> null,
                                        update -> {
                                            out.collect(update);
                                            return null;
                                        }
                                ),
                                heartbeat -> null
                        );
                    }
                })
                .keyBy(update -> invertedUserIds.get(update.getUserId()))
                .process(new UpdateProcessManual(conf.getRmiHost(), pageViewServiceName))
                .setParallelism(conf.getTotalUsers())
                .map(new TimestampMapper())
                .setParallelism(conf.getTotalUsers())
                .writeAsText(conf.getOutFile(), FileSystem.WriteMode.OVERWRITE);

        try {
            return env.execute("PageView Experiment");
        } finally {
            UnicastRemoteObject.unexportObject(pageViewService, true);
            try {
                registry.unbind(pageViewServiceName);
            } catch (final NotBoundException ignored) {

            }
        }
    }

    @Override
    public long getTotalEvents() {
        // PageView events + Get events + Update events
        return (conf.getTotalPageViews() * conf.getPageViewParallelism() +
                conf.getTotalPageViews() / 100 + conf.getTotalPageViews() / 1000) * conf.getTotalUsers();
    }

    @Override
    public long getOptimalThroughput() {
        return (long) ((conf.getPageViewParallelism() + 0.011) * conf.getTotalUsers() * conf.getPageViewRate());
    }

}

public class PageViewProcessManual extends
        KeyedBroadcastProcessFunction<Integer, PageViewOrHeartbeat, GetOrUpdateOrHeartbeat, Void> {

    private Instant pageViewPhysicalTimestamp = Instant.MIN;
    private Instant getOrUpdatePhysicalTimestamp = Instant.MIN;
    private int zipCode = 10_000;

    private final Queue<PageView> pageViews = new ArrayDeque<>();
    private final Queue<Update> updates = new ArrayDeque<>();

    private final String rmiHost;
    private final String pageViewServiceName;
    private final int pageViewParallelism;
    private int userId;
    private transient ForkJoinService<Integer, Integer> pageViewService;

    public PageViewProcessManual(final String rmiHost,
                                 final String pageViewServiceName,
                                 final int pageViewParallelism) {
        this.rmiHost = rmiHost;
        this.pageViewServiceName = pageViewServiceName;
        this.pageViewParallelism = pageViewParallelism;
    }

    @Override
    @SuppressWarnings("unchecked")
    public void open(final Configuration parameters) throws RemoteException, NotBoundException {
        final var registry = LocateRegistry.getRegistry(rmiHost);
        pageViewService = (ForkJoinService<Integer, Integer>) registry.lookup(pageViewServiceName);
        userId = getRuntimeContext().getIndexOfThisSubtask() / pageViewParallelism;
    }

    @Override
    public void processElement(final PageViewOrHeartbeat pageViewOrHeartbeat,
                               final ReadOnlyContext ctx,
                               final Collector<Void> out) throws RemoteException {
        pageViews.addAll(pageViewOrHeartbeat.match(List::of, hb -> Collections.emptyList()));
        pageViewPhysicalTimestamp = pageViewOrHeartbeat.getPhysicalTimestamp();
        makeProgress();
    }

    @Override
    public void processBroadcastElement(final GetOrUpdateOrHeartbeat getOrUpdateOrHeartbeat,
                                        final Context ctx,
                                        final Collector<Void> out) throws RemoteException {
        // We are getting all GetOrUpdate events, even if they have the wrong userId.
        // We filter out the wrong events.
        if (getOrUpdateOrHeartbeat.match(gou -> gou.getUserId() != userId, hb -> false)) {
            return;
        }
        updates.addAll(getOrUpdateOrHeartbeat.match(
                gou -> gou.match(g -> Collections.emptyList(), List::of),
                hb -> Collections.emptyList()));
        getOrUpdatePhysicalTimestamp = getOrUpdateOrHeartbeat.getPhysicalTimestamp();
        makeProgress();
    }

    private void makeProgress() throws RemoteException {
        final var currentTime = min(pageViewPhysicalTimestamp, getOrUpdatePhysicalTimestamp);
        while (!updates.isEmpty() &&
                updates.element().getPhysicalTimestamp().compareTo(currentTime) <= 0) {
            final var update = updates.remove();
            while (!pageViews.isEmpty() &&
                    pageViews.element().getPhysicalTimestamp().isBefore(update.getPhysicalTimestamp())) {
                process(pageViews.remove());
            }
            join(update);
        }
        while (!pageViews.isEmpty() &&
                pageViews.element().getPhysicalTimestamp().compareTo(currentTime) <= 0) {
            process(pageViews.remove());
        }
    }

    private void process(final PageView pageView) {
        // This update is a no-op
    }

    private void join(final Update update) throws RemoteException {
        zipCode = pageViewService.joinChild(getRuntimeContext().getIndexOfThisSubtask(), zipCode);
    }

}

public class UpdateProcessManual extends KeyedProcessFunction<Integer, Update, Update> {

    private final String rmiHost;
    private final String pageViewServiceName;
    private transient ForkJoinService<Integer, Integer> pageViewService;
    private int zipCode;

    public UpdateProcessManual(final String rmiHost, final String pageViewServiceName) {
        this.rmiHost = rmiHost;
        this.pageViewServiceName = pageViewServiceName;
        this.zipCode = 10_000;
    }

    @Override
    @SuppressWarnings("unchecked")
    public void open(final Configuration parameters) throws RemoteException, NotBoundException {
        final var registry = LocateRegistry.getRegistry(rmiHost);
        pageViewService = (ForkJoinService<Integer, Integer>) registry.lookup(pageViewServiceName);
    }

    @Override
    public void processElement(final Update update,
                               final Context ctx,
                               final Collector<Update> out) throws RemoteException {
        zipCode = pageViewService.joinParent(getRuntimeContext().getIndexOfThisSubtask(), update.zipCode);
        out.collect(update);
    }

}

public class PageViewService implements ForkJoinService<Integer, Integer> {

    private final int pageViewParallelism;
    private final List<List<Semaphore>> forkSemaphores;
    private final List<List<Semaphore>> joinSemaphores;
    private final List<Integer> zipCode;

    public PageViewService(final int totalUsers, final int pageViewParallelism) {
        this.pageViewParallelism = pageViewParallelism;
        forkSemaphores = new ArrayList<>(totalUsers);
        joinSemaphores = new ArrayList<>(totalUsers);
        zipCode = new ArrayList<>(totalUsers);
        for (int i = 0; i < totalUsers; ++i) {
            final var forkSems = new ArrayList<Semaphore>(pageViewParallelism);
            final var joinSems = new ArrayList<Semaphore>(pageViewParallelism);
            for (int j = 0; j < pageViewParallelism; ++j) {
                forkSems.add(new Semaphore(0));
                joinSems.add(new Semaphore(0));
            }
            forkSemaphores.add(forkSems);
            joinSemaphores.add(joinSems);
            zipCode.add(10_000);
        }
    }

    @Override
    public Integer joinChild(final int subtaskIndex, final Integer state) {
        final int parentId = subtaskIndex / pageViewParallelism;
        final int childId = subtaskIndex %
        joinSemaphores.get(parentId).get(childId).release();
        forkSemaphores.get(parentId).get(childId).acquireUninterruptibly();
        return zipCode.get(parentId);
    }

    @Override
    public Integer joinParent(final int subtaskIndex, final Integer state) {
        joinSemaphores.get(subtaskIndex).forEach(Semaphore::acquireUninterruptibly);
        zipCode.set(subtaskIndex, state);
        // We iterate over forkSemaphores backwards, to simulate children organized in a chain
        for (int i = forkSemaphores.get(subtaskIndex).size() - 1; i >= 0; --i) {
            forkSemaphores.get(subtaskIndex).get(i).release();
        }
        return state;
    }

}
\end{lstlisting}

\subsection{Fraud detection}

In this example, the input is completely analogous as in the Event-Based Windowing
example, with transactions playing the role of values, and rules playing the role
of barriers. The difference is in the computation: here the computation in one
window depends on the result of the previous window. Because of this, the parallel
implementation for Event-Based Windowing no longer works as there is no way to
propagate the necessary information from one window to another.

\subsubsection{Sequential implementation} The sequential implementation is
analogous to the sequential implementation of Event-Based Windowing.

\begin{lstlisting}[language=Java,basicstyle=\tiny\ttfamily,commentstyle=\tiny\ttfamily,morekeywords={var}]
public class FraudDetectionSequentialExperiment implements Experiment {

    private final FraudDetectionConfig conf;

    public FraudDetectionSequentialExperiment(final FraudDetectionConfig conf) {
        this.conf = conf;
    }

    @Override
    public JobExecutionResult run(final StreamExecutionEnvironment env, final Instant startTime) throws Exception {
        env.setParallelism(1);

        final var transactionSource = new TransactionOrHeartbeatSource(
                conf.getTotalValues(), conf.getValueRate(), startTime);
        final var transactionStream = env.addSource(transactionSource)
                .setParallelism(conf.getValueNodes())
                .slotSharingGroup("transactions");
        final var ruleSource = new RuleOrHeartbeatSource(
                conf.getTotalValues(), conf.getValueRate(), conf.getValueBarrierRatio(),
                conf.getHeartbeatRatio(), startTime);
        final var ruleStream = env.addSource(ruleSource)
                .slotSharingGroup("rules");

        ruleStream.connect(transactionStream)
                .process(new FraudProcessFunction(conf.getValueNodes()))
                .slotSharingGroup("rules")
                .map(new TimestampMapper())
                .writeAsText(conf.getOutFile(), FileSystem.WriteMode.OVERWRITE);

        return env.execute("FraudDetection Experiment");
    }

    @Override
    public long getTotalEvents() {
        return conf.getValueNodes() * conf.getTotalValues() + conf.getTotalValues() / conf.getValueBarrierRatio();
    }

    @Override
    public long getOptimalThroughput() {
        return (long) (conf.getValueRate() * conf.getValueNodes() + conf.getValueRate() / conf.getValueBarrierRatio());
    }

}

public class FraudProcessFunction extends
        CoProcessFunction<RuleOrHeartbeat, TransactionOrHeartbeat, Tuple3<String, Long, Instant>> {

    private final List<Instant> transactionTimestamps = new ArrayList<>();
    private final Tuple2<Long, Long> previousAndCurrentSum = Tuple2.of(0L, 0L);
    private final PriorityQueue<Transaction> transactions = new PriorityQueue<>(new TimestampComparator());
    private final Queue<Rule> rules = new ArrayDeque<>();
    private Instant ruleTimestamp = Instant.MIN;

    public FraudProcessFunction(final int transactionParallelism) {
        transactionTimestamps.addAll(Collections.nCopies(transactionParallelism, Instant.MIN));
    }

    @Override
    public void processElement1(final RuleOrHeartbeat ruleOrHeartbeat,
                                final Context ctx,
                                final Collector<Tuple3<String, Long, Instant>> out) {
        rules.addAll(ruleOrHeartbeat.match(List::of, hb -> Collections.emptyList()));
        ruleTimestamp = ruleOrHeartbeat.getPhysicalTimestamp();
        makeProgress(out);
    }

    @Override
    public void processElement2(final TransactionOrHeartbeat transactionOrHeartbeat,
                                final Context ctx,
                                final Collector<Tuple3<String, Long, Instant>> out) {
        transactions.addAll(transactionOrHeartbeat.match(List::of, hb -> Collections.emptyList()));
        transactionTimestamps.set(transactionOrHeartbeat.getSourceIndex(),
                transactionOrHeartbeat.getPhysicalTimestamp());
        makeProgress(out);
    }

    private Instant getCurrentTimestamp() {
        final var transactionTimestamp = transactionTimestamps.stream().min(Instant::compareTo).get();
        return min(transactionTimestamp, ruleTimestamp);
    }

    private void makeProgress(final Collector<Tuple3<String, Long, Instant>> out) {
        final var currentTimestamp = getCurrentTimestamp();
        while (!rules.isEmpty() &&
                rules.element().getPhysicalTimestamp().compareTo(currentTimestamp) <= 0) {
            final var rule = rules.remove();
            while (!transactions.isEmpty() &&
                    transactions.element().getPhysicalTimestamp().isBefore(rule.getPhysicalTimestamp())) {
                update(transactions.remove(), out);
            }
            update(rule, out);
        }
        while (!transactions.isEmpty() &&
                transactions.element().getPhysicalTimestamp().compareTo(currentTimestamp) <= 0) {
            update(transactions.remove(), out);
        }
    }

    private void update(final Rule rule,
                        final Collector<Tuple3<String, Long, Instant>> out) {
        out.collect(Tuple3.of("Rule", previousAndCurrentSum.f1, rule.getPhysicalTimestamp()));
        previousAndCurrentSum.f0 = previousAndCurrentSum.f1;
        previousAndCurrentSum.f1 = 0L;
    }

    private void update(final Transaction transaction,
                        final Collector<Tuple3<String, Long, Instant>> out) {
        if (previousAndCurrentSum.f0 %
            out.collect(Tuple3.of("Transaction", transaction.val, transaction.getPhysicalTimestamp()));
        }
        previousAndCurrentSum.f1 += transaction.val;
    }

}
\end{lstlisting}

\subsubsection{Manual implementation} Again, the manual implementation is
analogous to the manual implementation of Event-Based Windowing.

\begin{lstlisting}[language=Java,basicstyle=\tiny\ttfamily,commentstyle=\tiny\ttfamily,morekeywords={var}]
public class FraudDetectionManualExperiment implements Experiment {

    private final FraudDetectionConfig conf;

    public FraudDetectionManualExperiment(final FraudDetectionConfig conf) {
        this.conf = conf;
    }

    @Override
    public JobExecutionResult run(final StreamExecutionEnvironment env, final Instant startTime) throws Exception {
        env.setParallelism(1);

        // Set up the remote ForkJoin service
        final var fraudDetectionService = new FraudDetectionService(conf.getValueNodes());
        @SuppressWarnings("unchecked") final var fraudDetectionServiceStub =
                (ForkJoinService<Tuple2<Long, Long>, Tuple2<Long, Long>>)
                UnicastRemoteObject.exportObject(fraudDetectionService, 0);
        final var fraudDetectionServiceName = UUID.randomUUID().toString();
        final var registry = LocateRegistry.getRegistry(conf.getRmiHost());
        registry.rebind(fraudDetectionServiceName, fraudDetectionServiceStub);

        final var transactionSource = new TransactionOrHeartbeatSource(
                conf.getTotalValues(), conf.getValueRate(), startTime);
        final var transactionStream = env.addSource(transactionSource)
                .setParallelism(conf.getValueNodes())
                .slotSharingGroup("transactions");
        final var ruleSource = new RuleOrHeartbeatSource(
                conf.getTotalValues(), conf.getValueRate(), conf.getValueBarrierRatio(),
                conf.getHeartbeatRatio(), startTime);
        final var ruleStream = env.addSource(ruleSource)
                .slotSharingGroup("rules");

        final var broadcastStateDescriptor = new MapStateDescriptor<>("BroadcastState", Void.class, Void.class);
        final var broadcastRuleStream = ruleStream.broadcast(broadcastStateDescriptor);

        transactionStream.connect(broadcastRuleStream)
                .process(new TransactionProcessManual(conf.getRmiHost(), fraudDetectionServiceName))
                .setParallelism(conf.getValueNodes())
                .slotSharingGroup("transactions")
                .map(new TimestampMapper())
                .setParallelism(conf.getValueNodes())
                .writeAsText(conf.getTransOutFile(), FileSystem.WriteMode.OVERWRITE);

        ruleStream
                .flatMap(new FlatMapFunction<RuleOrHeartbeat, Rule>() {
                    @Override
                    public void flatMap(final RuleOrHeartbeat ruleOrHeartbeat, final Collector<Rule> out) {
                        ruleOrHeartbeat.match(
                                rule -> {
                                    out.collect(rule);
                                    return null;
                                },
                                heartbeat -> null
                        );
                    }
                })
                .process(new RuleProcessManual(conf.getRmiHost(), fraudDetectionServiceName))
                .map(new TimestampMapper())
                .writeAsText(conf.getOutFile(), FileSystem.WriteMode.OVERWRITE);

        try {
            return env.execute("FraudDetection Manual Experiment");
        } finally {
            UnicastRemoteObject.unexportObject(fraudDetectionService, true);
            try {
                registry.unbind(fraudDetectionServiceName);
            } catch (final NotBoundException ignored) {

            }
        }
    }

    @Override
    public long getTotalEvents() {
        return conf.getValueNodes() * conf.getTotalValues() + conf.getTotalValues() / conf.getValueBarrierRatio();
    }

    @Override
    public long getOptimalThroughput() {
        return (long) (conf.getValueRate() * conf.getValueNodes() + conf.getValueRate() / conf.getValueBarrierRatio());
    }

}

public class TransactionProcessManual extends BroadcastProcessFunction<TransactionOrHeartbeat, RuleOrHeartbeat,
        Tuple3<String, Long, Instant>> {

    private Instant transactionPhysicalTimestamp = Instant.MIN;
    private Instant rulePhysicalTimestamp = Instant.MIN;
    private Tuple2<Long, Long> previousAndCurrentSum = Tuple2.of(0L, 0L);

    private final Queue<Transaction> transactions = new ArrayDeque<>();
    private final Queue<Rule> rules = new ArrayDeque<>();

    private final String rmiHost;
    private final String fraudDetectionServiceName;
    private transient ForkJoinService<Tuple2<Long, Long>, Tuple2<Long, Long>> fraudDetectionService;

    public TransactionProcessManual(final String rmiHost, final String fraudDetectionServiceName) {
        this.rmiHost = rmiHost;
        this.fraudDetectionServiceName = fraudDetectionServiceName;
    }

    @Override
    @SuppressWarnings("unchecked")
    public void open(final Configuration parameters) throws RemoteException, NotBoundException {
        final var registry = LocateRegistry.getRegistry(rmiHost);
        fraudDetectionService =
                (ForkJoinService<Tuple2<Long, Long>, Tuple2<Long, Long>>) registry.lookup(fraudDetectionServiceName);
    }

    @Override
    public void processElement(final TransactionOrHeartbeat transactionOrHeartbeat,
                               final ReadOnlyContext ctx,
                               final Collector<Tuple3<String, Long, Instant>> out) throws RemoteException {
        transactions.addAll(transactionOrHeartbeat.match(List::of, hb -> Collections.emptyList()));
        transactionPhysicalTimestamp = transactionOrHeartbeat.getPhysicalTimestamp();
        makeProgress(out);
    }

    @Override
    public void processBroadcastElement(final RuleOrHeartbeat ruleOrHeartbeat,
                                        final Context ctx,
                                        final Collector<Tuple3<String, Long, Instant>> out) throws RemoteException {
        rules.addAll(ruleOrHeartbeat.match(List::of, hb -> Collections.emptyList()));
        rulePhysicalTimestamp = ruleOrHeartbeat.getPhysicalTimestamp();
        makeProgress(out);
    }

    private void makeProgress(final Collector<Tuple3<String, Long, Instant>> out) throws RemoteException {
        final var currentTime = min(transactionPhysicalTimestamp, rulePhysicalTimestamp);
        while (!transactions.isEmpty() &&
                transactions.element().getPhysicalTimestamp().compareTo(currentTime) <= 0) {
            final var transaction = transactions.remove();
            while (!rules.isEmpty() &&
                    rules.element().getPhysicalTimestamp().isBefore(transaction.getPhysicalTimestamp())) {
                join(rules.remove());
            }
            update(transaction, out);
        }
        while (!rules.isEmpty() &&
                rules.element().getPhysicalTimestamp().compareTo(currentTime) <= 0) {
            join(rules.remove());
        }
    }

    private void update(final Transaction transaction, final Collector<Tuple3<String, Long, Instant>> out) {
        if (previousAndCurrentSum.f0 %
            out.collect(Tuple3.of("Transaction", transaction.val, transaction.getPhysicalTimestamp()));
        }
        previousAndCurrentSum.f1 += transaction.val;
    }

    private void join(final Rule rule) throws RemoteException {
        previousAndCurrentSum =
                fraudDetectionService.joinChild(getRuntimeContext().getIndexOfThisSubtask(), previousAndCurrentSum);
    }

}

public class RuleProcessManual extends ProcessFunction<Rule, Tuple3<String, Long, Instant>> {

    private final String rmiHost;
    private final String fraudDetectionServiceName;
    private transient ForkJoinService<Tuple2<Long, Long>, Tuple2<Long, Long>> fraudDetectionService;

    private Tuple2<Long, Long> previousAndCurrentSum;

    public RuleProcessManual(final String rmiHost, final String fraudDetectionServiceName) {
        this.rmiHost = rmiHost;
        this.fraudDetectionServiceName = fraudDetectionServiceName;
        this.previousAndCurrentSum = Tuple2.of(0L, 0L);
    }

    @Override
    @SuppressWarnings("unchecked")
    public void open(final Configuration parameters) throws RemoteException, NotBoundException {
        final var registry = LocateRegistry.getRegistry(rmiHost);
        fraudDetectionService =
                (ForkJoinService<Tuple2<Long, Long>, Tuple2<Long, Long>>) registry.lookup(fraudDetectionServiceName);
    }

    @Override
    public void processElement(final Rule rule,
                               final Context ctx,
                               final Collector<Tuple3<String, Long, Instant>> out) throws RemoteException {
        previousAndCurrentSum = fraudDetectionService.joinParent(
                getRuntimeContext().getIndexOfThisSubtask(), previousAndCurrentSum);
        out.collect(Tuple3.of("Rule", previousAndCurrentSum.f1, rule.getPhysicalTimestamp()));
    }

}

public class FraudDetectionService implements ForkJoinService<Tuple2<Long, Long>, Tuple2<Long, Long>> {

    private final int transactionParallelism;
    private final List<Semaphore> joinSemaphores;
    private final List<Semaphore> forkSemaphores;
    private final List<Tuple2<Long, Long>> states;

    public FraudDetectionService(final int transactionParallelism) {
        this.transactionParallelism = transactionParallelism;
        joinSemaphores = IntStream.range(0, transactionParallelism)
                .mapToObj(x -> new Semaphore(0)).collect(Collectors.toList());
        forkSemaphores = IntStream.range(0, transactionParallelism)
                .mapToObj(x -> new Semaphore(0)).collect(Collectors.toList());
        states = IntStream.range(0, transactionParallelism)
                .mapToObj(x -> Tuple2.of(0L, 0L)).collect(Collectors.toList());
    }

    @Override
    public Tuple2<Long, Long> joinChild(final int subtaskIndex, final Tuple2<Long, Long> state) {
        states.set(subtaskIndex, state);
        joinSemaphores.get(subtaskIndex).release();
        forkSemaphores.get(subtaskIndex).acquireUninterruptibly();
        return states.get(subtaskIndex);
    }

    @Override
    public Tuple2<Long, Long> joinParent(final int subtaskIndex, final Tuple2<Long, Long> state) {
        long currentSum = 0L;
        for (int i = 0; i < transactionParallelism; ++i) {
            joinSemaphores.get(i).acquireUninterruptibly();
            currentSum += states.get(i).f1;
        }

        // We use the fact that all the children have the same previous sum, which otherwise we would
        // get from the state passed as an argument. As a consequence, we're not using state at all, and
        // the caller is free to make it null. Alternatively, we could have the children pass only their
        // current sums, and get the previous sum from state passed here as an argument.
        final long previousSum = states.get(0).f0;
        for (int i = transactionParallelism - 1; i >= 0; --i) {
            final var previousAndCurrentSum = states.get(i);
            previousAndCurrentSum.f0 = currentSum;
            previousAndCurrentSum.f1 = 0L;
            forkSemaphores.get(i).release();
        }
        return Tuple2.of(previousSum, currentSum);
    }

}
\end{lstlisting}

\end{document}